\documentclass [11pt]{article}

\usepackage{scribe}

\usepackage[backend=biber,style=alphabetic,backref=false]{biblatex}
\usepackage{xcolor}

\title{Finding Nearly-Periodic Components in Digraphs and Markov Chains from the Spectrum of Rotated Laplacian Matrices}

\author{
Salil Vadhan \thanks{\texttt{salil\_vadhan@harvard.edu }  Supported in part by a Simons Investigator Award}\\ Harvard University 
\and 
Jiyu Zhang\thanks{\texttt{jiyu.zhang@phd.unibocconi.it }  Part of the work was done while  visiting EPFL}\\ Bocconi University 
}

\date{}

\begin{document}

\maketitle

\begin{abstract}
     Inspired by recent advances in notions of spectral approximation of digraphs \cite{ahmadinejad2020high}, we study spectral algorithms for finding periodic structures in digraphs via the spectrum of a class of \emph{rotated} Laplacian matrices. This class of Laplacian matrices was previously studied by Lange, Liu, Peyerimhoff, and Post \cite{lange2015frustration}. We consider a notion of \emph{periodicity ratio} that generalizes the bipartiteness ratio of Trevisan \cite{Trevisan2009MaxCut}, and show that it is closely related to the spectrum of rotated Laplacian matrices. In particular, if the digraph is strongly connected and represents a Markov chain, this periodicity ratio for a given $p\in \N$ is a quantitative measure of how close this Markov chain is to having periodicity $p$.   
     
     We propose and analyze a periodicity-ratio variant of the spectral algorithm by Louis, Raghavendra, Tetali and Vempala \cite{LouisRaghavendraTetaliVempala12}. We show that the algorithm runs in randomized polynomial time and can find many nearly periodic components (i.e, components with small periodicity ratio). This also implies a new higher-order Cheeger-type inequality for periodicity in the spirit of that in \cite{LouisRaghavendraTetaliVempala12,lee2014multiway}.  
     
     As part of our analysis, we prove a new theorem that upper bounds the probability that the largest magnitudes of two sequences of coordinate-wise correlated complex Gaussian random variables occur at different indices, which may be of independent interest. Previously, an analogous result was known only for real Gaussian random variables. 
\end{abstract}

\section{Introduction}

Spectral graph theory studies the combinatorial properties of a graph by studying its associated matrices, most commonly the Laplacian matrix. For a graph $G=(V,E)$, the normalized Laplacian matrix associated with $G$ is the matrix $L= I- D^{-1/2} A D^{-1/2}$, where $D$ is the diagonal degree matrix and $A$ is the adjacency matrix. In this introduction, for simplicity we will mostly restrict our discussion to unweighted graphs. However, all results presented in this paper extend to general weighted graphs, and proofs are provided for the weighted case.   

A common theme in spectral graph theory is that the spectrum of the Laplacian matrix encodes the connectivity of the graph, and the eigenvectors constitute a geometric embedding of the vertices. Formally, we define the \emph{edge expansion} of a set $S\subseteq V$ as

\[\phi(S) = \frac{|E(S,V-S)|}{\sum_{v\in S} d_v}\]

where $E(S,V-S)$ is the set of edges between $S$ and $V-S$ and $d_v$ is the degree, i.e, the number of edges connected to the vertex $v$. The fundamental Cheeger's Inequality states that 

\begin{theorem} \cite{alon1985lambda1, alon1986eigenvalues}
    For a graph $G=(V=[n],E)$ , let $\lambda_1= 0 \leq \lambda_2 \leq \cdots \leq \lambda_n\leq 2$ be the eigenvalues of its associated normalized Laplacian matrix $L= I- D^{-1/2} A D^{-1/2}$. We have
    
    \[\lambda_2/2 \leq \min_{\substack{S\subseteq V, 
    |S|\leq |V|/2}} \phi(S) \leq \sqrt{2\lambda_2}\]
\end{theorem}

Cheeger's Inequality shows that the first non-trivial eigenvalue (i.e, $\lambda_2$) of the Laplacian matrix associated with a graph is a quantitative measure of the disconnectivity of the graph. In particular, if $\lambda_2 = 0$, then there exists a subset $S$ such that $\phi(S)=0$, meaning that the vertices in $S$ are disconnected with the rest of the graph. A line of work by Lee, Oveis-Gharan, Trevisan \cite{lee2014multiway} and Louis, Raghavendra, Tetali, Vempala \cite{LouisRaghavendraTetaliVempala12} has established higher-order Cheeger's Inequalities, relating higher eigenvalues of the Laplacian to the existence of partitions of the vertex set into many non-expanding pieces.  This line of work has provided a theoretical justification for the success of spectral clustering techniques \cite{ng2001spectral, von2007tutorial}, in which the spectral embedding of the vertices is computed in the first stage, then a clustering algorithm (e.g, the $k$-means algorithm) is applied to this embedding.   

In this work, we look at a class of \emph{rotated} Laplacian matrices, and show that their spectra enable us to find periodic structures in the graph. Formally, for a digraph $G=(V,E)$, its adjacency matrix is the $n\times n$ matrix $A$ where $A(v,u) = 1$ if there is an edge from $u$ to $v$, otherwise $0$. Let $A^*$ denote the conjugate transpose of $A$. For $p\in \N$ with $2\pi/p\in [0,\pi]$, let $z=e^{2\pi i/p}$ be a $p$-th root of unity, and let $z^*$ denote the complex conjugate of $z$.  We define $A_z = z^*A^* + zA$, which is the Hermitianization of the rotated adjacency matrix $zA$ (the entries in $A$ are rotated in the complex plane). The in-degree matrix of $G$, denoted by $D_{\text{in}}$,  is the diagonal matrix such that $D_{\text{in}}(v,v)$ is the number of edges entering $v$. Similarly, the out-degree matrix of $G$, denoted by $D_{\text{out}}$, is the diagonal matrix such that $D_{\text{out}}(v,v)$ is the number of edges leaving the vertex $v$. We let $D= D_{\text{in}} + D_{\text{out}}$, and $d_v = D(v,v)$. We define the \emph{normalized, $z$-rotated Laplacian matrix} associated with the graph $G$ as

\[L_z :=  D^{-1/2}\left(D-A_z\right)D^{-1/2} =  I- D^{-1/2}A_zD^{-1/2} \]

$L_z$ can be viewed as the generalization of $L$ by generalizing its definition using its incidence matrix. Recall that the incidence matrix of a graph $G=(V,E)$ is the $m\times n$ matrix $B$ where $m= |E|$ and $n=|V|$. The rows are indexed by the edges and the columns are indexed by the vertices. In the undirected case, for an edge $(u,v)$ we pick an arbitrary direction, say $u\to v$, then for the row indexed by $(u,v)$, put $1$ in the $u$-th entry and $-1$ in the $v$-th entry. It is well known that $L = D^{-1/2}B^TBD^{-1/2}$. In the directed case, we define an incidence matrix $B_z$ as follows: for the row indexed by the directed edge $u\to v$, put $1$ in the $u$-th entry and $-z^*$ in the $v$-th entry. It can be easily verified that $L_z = D^{-1/2}B_z^*B_zD^{-1/2}$  (see Proposition \ref{prop:Lz-property}).

This class of rotated Laplacian matrices was previously studied by Lange, Liu, Peyerimhoff, and Post \cite{lange2015frustration} under the name of graph Laplacians with cyclic signatures.  In their definition, each edge is associated with a signature in $\{z\in \C \mid z^p = 1\}$.  Our rotated Laplacian corresponds to the special case in which the same root of unity is assigned to every edge. Lange et al. show that the spectra of these graph Laplacians reveal cyclic  sub-structures of the graph. Li, Sun and Zanetti \cite{li2018hermitian} also studied similar Hermitian Laplacians associated with Max-2-Lin instances, and proved a Cheeger-type inequality for satisfiability of Max-2-Lin instances. Our motivation for studying these rotated Laplacian matrices comes from recent advances in spectral approximation of digraphs and random-walk matrices. In particular, Ahmadinejad et al. \cite{ahmadinejad2020high} studied the matrix of $L_z$ in the special case of regular digraphs (in which case $D^{-1/2}AD^{-1/2}$ is the doubly stochastic random-walk matrix). In their work, they call a regular digraph $\widetilde{G}$ a \emph{unit-circle approximation} of a regular digraph $G$ if for all $z\in \C, |z|=1$, we have $(1-\epsilon)L_z \preceq \widetilde{L}_z \preceq (1+\epsilon)L_z$, where $\preceq$ denotes the Loewner order. The benefit of this notion is that the approximation is preserved under powering the random walk matrix. The intuition is that approximation of $L_z$ for all $z$ preserves the bipartite and periodic structures in $G$. 

We define and work with a notion of \emph{periodicity ratio} which can be viewed as a special case of the Cheeger constant for graph Laplacians with cyclic signatures as in \cite{lange2015frustration}. We choose to work specifically with periodicity ratio because our focus is on the periodicity of the Markov chain associated with the graph. We start with the definition of a $p$-periodic digraph:

\begin{definition}
    For $p\in \N$, a digraph $G=(V,E)$ is \emph{$p$-periodic} if its vertex set $V$ can be partitioned into $p$ disjoint subsets $C_0, \ldots, C_{p-1}$ such that every directed edge $(u,v)$ has $u\in C_{i}$ and $v\in C_{i+1 (\text{mod } p)}$ for some $i$.
\end{definition} 

See Figure \ref{fig:figure1} for an example of a $4$-periodic graph. The connection between the periodic structure and the spectrum of $L_z$ can be seen from the characterization of its eigenvalues using Rayleigh quotient (see Section \ref{sec:prelim}). By the Courant-Fischer theorem (Theorem \ref{thm:CF-eigen}), the first eigenvalue of $L_z$ is characterized by

\[\lambda_1(L_z)  = \min_{u\in \C^n, u\neq \mathbf{0}} R_{L_z}(D^{1/2}u)  = \min_{u\in \C^n, u\neq \mathbf{0}}\frac{u^*B_z^*B_zu}{u^*Du}  = \min_{u\in \C^n, u\neq \mathbf{0}} \frac{\sum_{w\to v} |u_w-z\cdot u_v|^2}{\sum_vd_v|u_v|^2}\]

\begin{figure}[htbp]
    \centering

    \begin{minipage}[b]{0.35\textwidth}
        \centering
        \includegraphics[width=\textwidth]{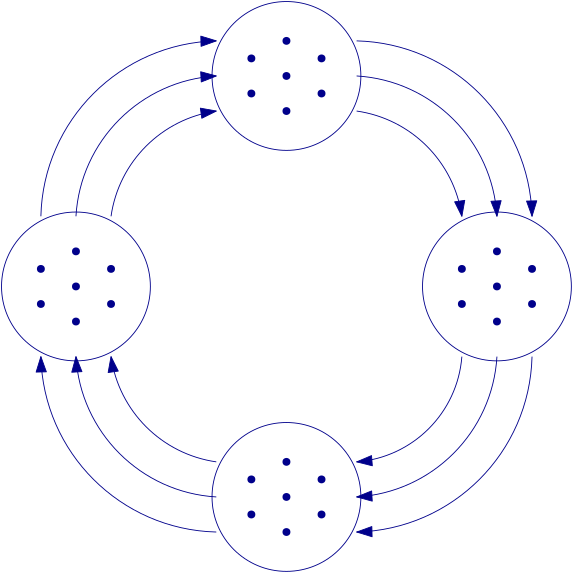}
        \caption{4-periodic component}
        \label{fig:figure1}
    \end{minipage}
    \hfill
    \begin{minipage}[b]{0.42\textwidth}
        \centering
        \includegraphics[width=\textwidth]{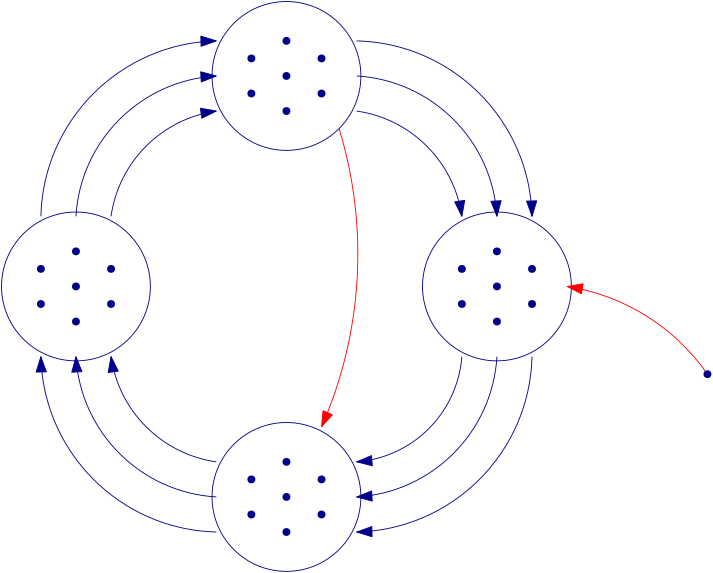}
        \caption{Nearly 4-periodic component with error edges}
        \label{fig:figure2}
    \end{minipage}

\end{figure}

When $p\geq 2$ and the graph $G$ contains a $p$-periodic component, $\lambda_1(L_z)$ is $0$ and the associated test eigenvector $x$ assigns the value $z^{-i}$ to vertices in $C_i$ and the value $0$ to vertices outside  the component. This motivates us to define a notion that captures the approximate periodicity of $G$. Analogous to the bipartiteness ratio introduced by Trevisan \cite{Trevisan2009MaxCut}, we define the \emph{$p$-periodicity ratio} of a graph $G$ as follows:

\begin{definition}
    Let $G=(V,E)$ be a digraph, $p\in \N$ and $z= e^{2\pi i/p}$. For a complex vector $y \in \{0, z^0, z^1 ,\ldots, z^{p-1}\}^n$, the \emph{$p$-periodicity ratio} of $y$ in $G$ is defined as
    \[\mathcal{P}_p(y):= \frac{\sum_{u\to v}|y_u-z\cdot y_v|}{\sum_vd_v|y_v|} \]

    The \emph{$p$-periodicity ratio} of $G$ is

    \[\mathcal{P}_p(G) := \min_{\substack{y \in \{0, z^0, z_1 ,\ldots, z^{p-1}\}^n \\ y\neq \mathbf{0}}} \mathcal{P}_p(y)\]
\end{definition} 

The quantity $\cP_p(G)$ also has a combinatorial interpretation. It can be viewed as an approximation of the minimum number of edges whose removal creates a $p$-periodic component $S$, relative to the number of edges incident to $S$ (i.e,  the volume of $S$). Formally, for  
$S\subseteq V$, define
\[\erredge(S) := \min_{F\subseteq E} \left\{|F|: S\text{ is a } p\text{-periodic component of }G'=(V, E\setminus F) \right\}\]
Let $\vol(S)=\sum_{v\in S} d_v$ be the total number of  edges entering or leaving vertices in $S$, then we have 
\[ \frac{\cP_p(G)}{2}\leq \min_{S\subseteq V}\frac{\erredge(S)}{\vol(S)} \leq \frac{\cP_p(G)}{2\sin(\pi/p)}  \]

More generally, for $\emptyset \subsetneq S\subseteq V$, we have

\[\frac{\left(\min_{\substack{y \in \{0, z^0, z_1 ,\ldots, z^{p-1}\}^n \\ \supp(y)=S}} \mathcal{P}_p(y)\right)}{2}\leq \frac{\erredge(S)}{\vol(S)} \leq \frac{\left(\min_{\substack{y \in \{0, z^0, z_1 ,\ldots, z^{p-1}\}^n \\ \supp(y)=S}} \mathcal{P}_p(y)\right)}{2\sin(\pi/p)}\]

Figure \ref{fig:figure2} shows an example of a nearly $4$-periodic component with error edges. See Section \ref{sec:periodicity} for the statement and proof of the general weighted case. 

Lange et al. \cite{lange2015frustration} proved Cheeger-type inequalities and higher-order Cheeger-type inequalities for graph Laplacians with cyclic signatures. As a result, the following Cheeger-type inequality and higher-order Cheeger-type inequality for periodicity are shown in their work:

\begin{theorem}\label{thm:cheeger-unified}
    (\cite{lange2015frustration}, Theorem 4.1 and Theorem 5.1) For a digraph $G=(V,E)$, let $p\in \N$ and $z = e^{2\pi i/p}$. Let $L_z =  I - D^{-1/2}A_zD^{-1/2}$ be its associated rotated Laplacian matrix with eigenvalues $0\leq \lambda_1(L_z) \leq \cdots \leq \lambda_n(L_z)\leq 2$, then

    \[\lambda_1(L_z)/2 \leq \mathcal{P}_p(G) \leq 2\sqrt{2\lambda_1(L_z)}\]
    and for every $k\in \N$

    \[\lambda_k(L_z)/2\leq \min_{\substack{\text{ disjointly supported} \\y_{1}, \ldots, y_{k}\in  \{0, z^0, \ldots, z^{p-1}\}^n }} \max_{i\in[k]} \mathcal{P}_p(y_{i}) \leq O(k^3) \cdot \sqrt{\lambda_{k}(L_z)} \]
    
\end{theorem}

The first pair of inequalities are a generalization of Trevisan's Cheeger-type inequality for bipartiteness, which corresponds to the case of $p=2$. The right-hand side of the first pair of inequalities is proven by a threshold rounding argument developed in \cite{lange2015frustration}. In particular, this implies a Cheeger-type rounding algorithm for finding a periodic component in the graph, which can be viewed as a generalization of Trevisan's algorithm \cite{Trevisan2009MaxCut}. We give an exposition of the Cheeger-type rounding algorithm (Algorithm \ref{algo:cheeger}) and its analysis in Appendix \ref{sec:cheeger}, as we  use it as a subroutine in our main algorithm (Algorithm \ref{algo:main}).  The second pair of inequalities are a periodicity analogue of the higher-order Cheeger's Inequality of Lee, Oveis-Gharan and Trevisan \cite{lee2014multiway}. Like \cite{lee2014multiway}, this proof is \emph{not} accompanied by an explicit polynomial-time algorithm.

\subsection{A spectral algorithm for finding many periodic components via higher eigenvalues of the rotated Laplacian}

Our main result is a spectral algorithm for finding many disjoint, nearly periodic components (Algorithm \ref{algo:main}). It is a rotated-Laplacian variant of the algorithm of Louis et al.  \cite{LouisRaghavendraTetaliVempala12}. In particular, the algorithm implies a new higher-order Cheeger-type inequality for periodicity in the spirit of that in \cite{LouisRaghavendraTetaliVempala12}.

\begin{theorem}\label{thm:main}
    There exists a constant $c\in (0,1)$ such that the following holds: for a digraph $G=(V,E)$, $p\in \mathbb{N}$ and $z= e^{2\pi i/p}$. Let $L_z =  I - D^{-1/2}A_zD^{-1/2}$ be its associated rotated Laplacian matrix with eigenvalues $0\leq \lambda_1(L_z) \leq \cdots \leq \lambda_n(L_z)\leq 2$, then we have
    \[\min_{\substack{\text{ disjointly supported} \\y_{1}, \ldots, y_{ck}\in  \{0, z^0, \ldots, z^{p-1}\}^n }} \max_{i\in \left[\lceil ck\rceil\right]} \mathcal{P}_p(y_{i}) \leq O(\sqrt{\log k}) \cdot \sqrt{\lambda_{k}(L_z)} \]
\end{theorem}

Compare to the second pair of inequalities in Theorem \ref{thm:cheeger-unified}. Theorem \ref{thm:main} obtains an improved upper bound of $O(\sqrt{\log k})\cdot \sqrt{\lambda_k(L_z)}$ instead of $O(k^3)\cdot \sqrt{\lambda_k(L_z)}$, at the price of identifying a constant factor fewer disjointly supported vectors.

Informally, the algorithm proceeds in three stages: in the first stage, the algorithm computes the spectral embedding of the vertices using the first $k$ eigenvectors of the rotated Laplacian matrix. In the second stage, it constructs $k$ disjointly supported vectors by applying Gaussian projection to the embedding vectors. In the last stage, the Cheeger-type rounding algorithm (Algorithm \ref{algo:cheeger}) is applied to each of these vectors to produce disjoint components. Our analysis shows that each component has periodicity ratio at most $O(\sqrt{\log k})\cdot \sqrt{\lambda_k(L_z)}$ with constant probability, yielding the following algorithmic result of Theorem \ref{thm:main}:

\begin{theorem}\label{thm:main-algo}
    There exists a randomized polynomial-time spectral algorithm and a constant $0<c<1$ such that, given a digraph $G=(V,E)$ and $k, p\in \N $ as input, the algorithm outputs with constant probability a collection of disjointly supported vectors $y_{1}, \ldots, y_{ck} \in  \{0, z^0, \ldots, z^{p-1}\}^n $ with periodicity ratio bounded by  
    \[\max_{i\in [\lceil ck \rceil]} \mathcal{P}_p(y_{i}) \leq O(\sqrt{\log k}) \cdot \sqrt{\lambda_{k}(L_z)}\]
    where $\lambda_k(L_z)$ is defined in Theorem \ref{thm:main}.
\end{theorem}

At a high level, our algorithm is guided by the following intuition: let $\{F_v \}_{v\in V}\subseteq \C^k$ be the spectral embedding vectors of the vertices using the first $k$ eigenvectors of $L_z$. For vertices from the same cluster $S$, the corresponding embedding vectors are expected to be close to a common one-dimensional subspace $\operatorname{span}(g)$ for $g\in \C^k, \|g\| =1$. Moreover, the label information within $S$ is encoded by the 
\emph{phase} of the complex projection $\ip{F_v}{g}$ (see Section \ref{sec:algo-main}).

In the process of proving Theorems \ref{thm:main} and \ref{thm:main-algo}, we will need to analyze the probability that the largest magnitudes of two sequences of coordinate-wise correlated complex Gaussian random variables occur at different indices. In particular, we prove the following upper bound on this probability, which generalizes a result for real Gaussian random variables by Charikar, Makarychev and Makarychev (\cite{charikar2006near}, Theorem 4.1), and may be of independent interest.

\begin{theorem}

    (Restatement of Theorem \ref{fact:correlation})
    Let $(X_1,Y_1), \ldots, (X_k, Y_k)$ be $k$  independent pairs of complex random variables, where $(X_i, Y_i)$ are $\rho_i $-correlated complex Gaussians.  If the average absolute covariance of $\{X_i\}$ and $\{Y_i\}$ is at least $1-\epsilon$, that is,

    \[\frac{\sum_i |\rho_i|}{k} \geq 1-\epsilon\]

    then 
    
    \[\pr[\arg\max_i |X_i| \neq \arg\max_i |Y_i|] \leq c_1 \sqrt{\epsilon\log k}\]
    for some absolute constant $c_1>0$.
\end{theorem}

The result for real Gaussian random variables was used in analyzing their algorithm for Unique Games. This result was later used in the analysis of the algorithm of  Louis et al. \cite{LouisRaghavendraTetaliVempala12}. Our proof is inspired by \cite{charikar2006near} and follows a similar strategy.

Finally, as part of the analysis of our algorithm, we also fill in a gap \footnote{The gap was a missing proof of a non-trivial inequality. The author acknowledged that they do not have a proof of it \cite{louis-private}. However, in Louis' thesis \cite{louis2014complexity}, an alternative algorithm is provided, whose analysis does not rely on this inequality.} in the original proof of Louis et al.  \cite{LouisRaghavendraTetaliVempala12}, with 7 pages of proof in Appendix \ref{ap:prop5.13}, including ChatGPT-assisted calculations.

\subsection{Organization}

In Section \ref{sec:prelim} we introduce necessary notations and background. In Section \ref{sec:periodicity} we give a combinatorial interpretation of the periodicity ratio. In Section \ref{sec:algo-main} we present our algorithm for finding many nearly periodic components and in Section \ref{sec:analysis-main} we present its analysis. In Appendix \ref{sec:cheeger}, we provide an exposition of the Cheeger-type rounding algorithm for finding periodic component and its analysis. Deferred proofs are provided in Appendices \ref{ap:prelim}, \ref{ap:prop5.13}, \ref{appendix}.

\section{Preliminaries}\label{sec:prelim}

We will work with vector spaces over the field $\C$ of complex numbers. For a complex number $a$, we denote the real part of $a$ by $\text{Re}(a)$ and the imaginary part by $\text{Im}(a)$. The magnitude of $a$ is $|a|:=\sqrt{a^*a}$. For complex vectors $x, y\in \C$, their inner product is $\ip{x}{y}:= x^*y$.   The $\ell_2$ norm or the length of a vector $x\in \C^n$ is $\|x\| := \sqrt{\ip{x}{x}}$.  We denote $\bar{x} = x/\|x\|$ as the normalized unit vector of $x$. For $n\in \N$ we use $[n]$ to denote the set $\{1,\ldots, n\}$. We use $\I(\mathcal{A})$ to denote the indicator random variable for the event $\mathcal{A}$.

In an undirected graph we use $u\sim v$ to represent the relation that $u$ is adjacent to $v$. In a directed graph we use $u\to v$ to indicate there is an edge from $u$ to $v$. We use an ordered pair $(u,v)$ to represent the directed edge $u\to v$.  For a weighted graph $G=(V,E)$, a weight function $w:V\times V \to \R_{\geq 0}$ is associated such that each edge $(u,v)\in E\subseteq V\times V$ has weight $w(u,v)$. If the graph is undirected, we have $w(u,v) = w(v,u)$.  

For a weighted graph $G=(V,E)$, the adjacency matrix $A$ has $A(v,u) =w(u,v)$ if there is an edge $u\to v$ with weight $w(u,v)$. The in-degree matrix $D_{\text{in}}$ is a diagonal matrix with $D_{\text{in}}(v,v) = \sum_{u\to v} w(u,v)$. Similarly, the out-degree matrix $D_{\text{out}}$ is a diagonal matrix with $D_{\text{out}}(v,v) = \sum_{v\to u} w(v,u)$. The incidence matrix $B_z$ is the $m\times n$ matrix with $m=|E|, n=|V|$, such that the row indexed by $(u,v)$ has $\sqrt{w(u,v)}$ in the $u$-th entry and $(-z^*)\cdot \sqrt{w(u,v)}$ in the $v$-th entry.

The following are basic properties of rotated Laplacian matrices. Their proofs are deferred to Appendix \ref{ap:prelim}.
\begin{proposition}\label{prop:Lz-property}  
    For $z=e^{2\pi i/d}$ for some $ d\in \mathbb{N}$, the normalized rotated Laplacian matrix is defined as $L_z = I-D^{-1/2}A_zD^{-1/2}$ and has the following properties:
    \begin{enumerate}
        \item $L_z = D^{-1/2}B_z^*B_zD^{-1/2}$
        \item $L_z$ is Hermitian, positive semidefinite, and has eigenvalues $0\leq \lambda_1(L_z) \leq \cdots \leq \lambda_n(L_z)\leq 2$.
    \end{enumerate}
\end{proposition}

The \emph{Rayleigh quotient} of $x\in \C^n$ with respect to $L_z$ is 

\[R_{L_z}(x): = \frac{x^*L_zx}{x^*x}\]

Substituting $x$ using $D^{1/2}x$, the Rayleigh quotient of $D^{1/2}x$ with respect to $L_z$ has the form  

\[R_{L_z}(D^{1/2}x) = \frac{\sum_{u\to v} w(u,v)\cdot|x_u- z\cdot x_v|^2}{\sum_v d_v |x_v|^2}\]

The Courant-Fischer theorem provides a min-max characterization of the eigenvalues of $L_z$ using the Rayleigh quotient as an objective function: 
\begin{theorem}\label{thm:CF-eigen}
    (Courant-Fischer) For $1\leq k\leq n$, the eigenvalue $\lambda_k(L_z)$ of $L_z$ has the characterization
    \begin{equation}\label{eq:rayleigh}
        \lambda_k(L_z) = \min_{\substack{k-dim\\ \text{subspace } \mathcal{V}}}\max_{x \in \mathcal{V}-\{\mathbf{0}\}} R_{L_z}(x) = \min_{\substack{k-dim\\ \text{subspace } \mathcal{V}}} \max_{x \in \mathcal{V}-\{\mathbf{0}\}} R_{L_z}(D^{1/2}x)
    \end{equation}
\end{theorem}

We also recall the one-sided Chebyshev's inequality

\begin{fact}\label{fact:cherbyshev} (\cite{boucheron2003concentration}])
    Let $X$ be a real random variable, for any $t>0$, 

    \[\pr\left[X<\E[X] - t\cdot \sqrt{\mathrm{Var}[X]} \right] \leq 
    \frac{1}{1+ t^2}\]
\end{fact}

\section{Periodicity ratio and its combinatorial interpretation}\label{sec:periodicity}

In this section, we define periodicity ratio for weighted graphs, and show that there is a combinatorial interpretation of the periodicity ratio. 

\begin{definition}
    Let $G=(V,E)$ be a digraph with weight function $w:V\times V\to \R_{\geq0}$, $p\in \N$ and $z= e^{2\pi i/p}$. For a complex vector $y \in \{0, z^0, z^1 ,\ldots, z^{p-1}\}^n$, the \emph{$p$-periodicity ratio} of $y$ in $G$ is defined as
    \[\mathcal{P}_p(y):= \frac{\sum_{u\to v} w(u,v)\cdot |y_u-z\cdot y_v|}{\sum_vd_v|y_v|} \]

    the $p$-periodicity ratio of $G$ is

    \[\mathcal{P}_p(G) = \min_{\substack{y \in \{0, z^0, z_1 ,\ldots, z^{p-1}\}^n \\ y\neq \mathbf{0}}} \mathcal{P}_p(y)\]
\end{definition} 

Here, we show $p$-periodicity ratio approximately captures the minimum total weight of an edge set whose removal creates a $p$-periodic component $S$, relative to the total weight of the edges incident to $S$ (i.e, the volume of $S$). Formally, for  
$S\subseteq V$, define $\erredge(S)$ as
\[\erredge(S) := \min_{F\subseteq E} \left\{\sum_{(u,v)\in F}w(u,v): S\text{ is a } p\text{-periodic component of }G'=(V, E\setminus F) \right\}\]
Let $\vol(S)=\sum_{v\in S} d_v$ be the total weight of the edges entering or leaving vertices in $S$, then the following holds:

\begin{proposition} 
    For $\emptyset \subsetneq S\subseteq V$, we have

\[\frac{\left(\min_{\substack{y \in \{0, z^0, z_1 ,\ldots, z^{p-1}\}^n \\ \supp(y)=S}} \mathcal{P}_p(y)\right)}{2}\leq\frac{\erredge(S)}{\vol(S)} \leq \frac{\left(\min_{\substack{y \in \{0, z^0, z_1 ,\ldots, z^{p-1}\}^n \\ \supp(y)=S}} \mathcal{P}_p(y)\right)}{2\sin(\pi/p)}\]

In particular,

    \[ \frac{\cP_p(G)}{2}\leq \min_{S\subseteq V}\frac{\erredge(S)}{\vol(S)} \leq \frac{\cP_p(G)}{2\sin(\pi/p)}  \]
    
\end{proposition}

\begin{proof}

We first focus on proving the first pair of inequalities.

Fix a set $S$ such that $\emptyset \not\subseteq S \subseteq V$.  For the lower bound, let $F$ be a set of error edges which witnesses $\erredge(S)$, i.e, $\erredge(S) = \sum_{(u,v)\in F}w(u,v)$, and $S$ is a $p$-periodic component in $G'=(V, E\setminus F)$.   Partition $S$ into disjoint subsets $C_0, \ldots, C_{p-1}$ such that there only exist edges from $C_i$ to $C_{(i+1) \text{ mod } p}$ in $G'$. Let $y\in \{0, z^0, z^1, \ldots, z^{p-1}\}^n$ be the vector with $y_v = z^{-i}$ if $v\in C_i$ for some $i\in [p]$, otherwise $y_v=0$. Note that if $(u,v)\in E\setminus F$ then $|y_u - z\cdot y_v| = 0$.   

$F$ consists of two types of error edges.  A jumping edge $u\to v$ for $u,v\in S$ is an edge that leaves $u\in C_i$ and enters $v\in C_j$ for $j\neq i+1  ~(\text{mod } p)$. The value of  $w(u,v)\cdot |y_u-z\cdot y_v|$ is at most $2w(u,v)$. A cut edge $u\to v$ is an edge such that one end point of $(u,v)$ is in $S$ while the other is not. For such an edge $w(u,v)\cdot |y_u-z\cdot y_v| = w(u,v)$. Adding together, and by definition $\cP_p(G) = \min \cP_p(y)$  for all $y\in \{0, z^0, z^1, \ldots, z^{p-1}\}^n\setminus \{\mathbf{0}\}$, we have

\[\cP_p(G) \leq \min_{\substack{y \in \{0, z^0, z_1 ,\ldots, z^{p-1}\}^n \\ \supp(y)=S}} \mathcal{P}_p(y)\leq  \cP_p(y)\leq \frac{2\erredge(S)}{\mathrm{Vol}(S)}\]

For the upper bound, consider the vector $y\in \{0, z^0, z^1, \ldots, z^{p-1}\}^n$ with $\supp(y) = S$ such that 
\[\cP_p(y)=\min_{\substack{y \in \{0, z^0, z_1 ,\ldots, z^{p-1}\}^n \\ \supp(y)=S}} \mathcal{P}_p(y)\]

We will show that  $\erredge(S) \leq \mathrm{Vol}(S) \cdot \cP_p(y)/2\sin(\pi/p)$. Define 

\[F:=\{(u,v)\in E : y_u \neq z\cdot y_v\}\]

Note that by definition $S$ is a $p$-periodic component in $G' = (V,E\setminus F)$, witnessed by the sets $C_i= \{v: y_v=z^{-i}\}$ for $i\in [p]$.
We argue that for every $(u,v)\in F$, $|y_u-z\cdot y_v| \geq 2\sin (\pi/p) $.  Note that $|1-e^{i\alpha}| = 2\sin(\alpha/2)$, which is monotonically increasing for $\alpha\in [0,\pi]$ and decreasing for $\alpha\in [\pi,2\pi]$. Therefore $|y_u - z\cdot y_v| = |1-z^j|$ for some $j\in \{1,\ldots, p-1\}$ is minimized at $|1-z| = 2\sin (\pi/p)$, so $ |y_u-z\cdot y_v| \geq 2\sin(\pi/p)$. Thus, we have

\begin{align*}
    \erredge(S)&\leq \sum_{(u,v)\in F} w(u,v)\\
    &\leq  \frac{\sum_{(u,v)\in F}w(u,v)\cdot |y_u-z\cdot y_v|}{2\sin (\pi/p)}\\
    &\leq \frac{\vol(S)\cdot \cP_p(y)}{2\sin(\pi/p)}
\end{align*}

Now we show the second pair of inequalities can be proved using the above arguments. For the lower bound we take a set $S$ that minimizes the ratio $\erredge(S)/\vol(S)$ over all $\emptyset \not\subseteq S\subseteq V$. Our lower bound argument above shows that there exists a vector $y\in \{0, z^0, z^1, \ldots, z^{p-1}\}^n$  with $\supp(y) =S$  such that 

\[\cP_p(G)\leq  \cP_p(y)\leq \frac{2\erredge(S)}{\mathrm{Vol}(S)}\]

For the upper bound, we consider the vector $y\in \{0, z^0, z^1, \ldots, z^{p-1}\}^n$ that minimizes the periodicity ratio, i.e, $\cP_p(G)= \cP_p(y)$. Let $Q = \supp(y)$. Our upper bound argument above  shows that 

\[\min_{S\subseteq V}\frac{\erredge(S)}{\vol(S)} \leq \frac{\erredge(Q)}{\vol(Q)} = \frac{ \cP_p(G)}{2\sin(\pi/p)}\]

\end{proof}

\section{A spectral algorithm for finding many nearly-periodic components}\label{sec:algo-main}

\begin{theorem}
    There exists a randomized polynomial-time spectral algorithm and a constant $0<c<1$ such that, given a digraph $G=(V,E)$ and $k, p\in \N $ as input, the algorithm outputs with constant probability a collection of disjointly supported vectors $y_{1}, \ldots, y_{\lceil ck\rceil} \in  \{0, z^0, \ldots, z^{p-1}\}^n $ with periodicity ratio bounded by  
    \[\max_{i\in [\lceil ck\rceil]} \mathcal{P}_p(y_{i}) \leq O(\sqrt{\log k}) \cdot \sqrt{\lambda_{k}(L_z)}\]
    where $\lambda_k(L_z)$ is defined in Theorem \ref{thm:main}.
    
\end{theorem}

To describe our main algorithm, we first introduce spectral embedding and Gaussian random variables. We then present a result of Lange et al. \cite{lange2015frustration} which establishes a Cheeger-type rounding algorithm for finding a nearly periodic component in a graph using the first eigenvector of its rotated Laplacian matrix. Building on these ingredients, we formally describe our spectral algorithm for finding many periodic components in Algorithm \ref{algo:main}. 

\subsection*{Spectral embedding}

Let $u^{(1)}, \ldots, u^{(k)} \in \C^n$ be the $k$ orthonormal eigenvectors associated with the $k$ smallest eigenvalues of $L_z$. Set $x^{(i)} := D^{-1/2}u^{(i)}$ such that $\|D^{1/2}x^{(i)}\| = 1$ for each $i\in[k]$. We define the spectral embedding of $V$ as the collection of $n$ vectors $\{F_v\}_{v\in V}$ defined as

\[F_v := (x_v^{(1)},  \ldots,  x_v^{(k)})\]

We will sometimes refer to $F$ as the matrix with $F_v$'s as its rows, or equivalently the matrix with $x^{(i)}$'s as its columns. Define the higher-order Rayleigh quotient of $D^{1/2}F$ with respect to $L_z$  asa

\[R_{L_z}(D^{1/2}F) :=  \frac{\sum_{u\to v} w(u,v)\cdot \|F_u -z\cdot  F_v\|^2}{\sum_v d_v \cdot \|F_v\|^2}\]

\subsection*{Gaussian random variables}
A standard Gaussian random variable is denoted by $Z\sim \gau(0,1)$, and a standard complex Gaussian random variable can be viewed as a two dimensional real Gaussian random variable, denoted by $ Z \sim \cC \gau (0,1)$, so that $Z\sim \cC\gau(0,1)$ is equivalent to $Z = \frac{1}{\sqrt{2}}(X + iY)$ for independent $X,Y \sim \gau(0,1)$. We use $\re(Z)=X$ and $\im(Z)=Y$ to denote its real part and imaginary part. A standard complex Gaussian vectors, denoted by $g:=(g_1,\ldots, g_k)\sim \cC\gau(0,I_k)$ is obtained by taking $k$ independent copies of standard complex Gaussian random variables, i.e, $\forall i\in [k], g_i \sim \cC\gau(0,1)$ .

\subsection*{Cheeger-type rounding algorithm of Lange et al.}

As mentioned in the introduction, a Cheeger-type rounding algorithm can be extracted from the work of Lange et al. \cite{lange2015frustration}. We provide an exposition of this algorithm and its analysis in Appendix \ref{sec:cheeger}. Here we will state the following lemma for the purpose of presenting our main algorithm and its analysis.

\begin{lemma}\label{lem:lange}
    (\cite{lange2015frustration}) For a digraph $G=(V,E)$, $p\in \N, z = e^{2\pi i/ p}$, given a complex vector $x\in \C^n$, 
    there exists a polynomial-time algorithm (Algorithm \ref{algo:cheeger}) that rounds $x$ into a vector $y\in \{0,z^0,z^1,\ldots, z^{p-1}\}^n$ with
    \[\cP_p(y) \leq 2\cdot \frac{\sum_{u\to v} w(u,v)\cdot |x_u-z\cdot x_v|\cdot (|x_u| + |x_v|)}{\sum_v d_v|x_v|^2} \]
    Moreover, $\supp(y)\subseteq \supp(x)$.
\end{lemma}

We are now ready to present our main algorithm.

\subsection{Main algorithm}

\begin{algorithm}[H]
\caption{Finding many periodic components}
\label{algo:main}
\begin{algorithmic} 
\State Input: graph $G=(V,E)$, integers $k, p\in \N$.

\begin{enumerate}
    \item  \textbf{Spectral embedding:} 
    \begin{itemize}
        \item compute the $k$ orthonormal eigenvectors $u^{(1)} ,\ldots, u^{(k)}$ corresponding to the $k$ smallest eigenvalues of the normalized rotated Laplacian $L_z =  I - D^{-1/2}A_zD^{-1/2}$ for $z=e^{2\pi/p}$.
        \item rescale each eigenvector by setting $x^{(i)} \leftarrow D^{-1/2}u^{(i)}$.
        \item set $F_v := (x_v^{(1)}, \ldots, x_v^{(k)})$. 
    \end{itemize}
    
    \item \textbf{Gaussian projection:} 
    \begin{itemize}
        \item pick $k$ independent complex Gaussian vectors $g_1, \ldots, g_k \sim \mathcal{C}\gau(0,I_k)$.
        \item construct $h^{(i)}$ for each direction $g_i, i\in [k]$ as
        \[ h_v^{(i)} = 
            \begin{cases}
            \ip{F_v}{g_i} & \text{if } |\ip{F_v}{g_i}| > |\ip{F_v}{g_j}| ~ \forall j\neq i \\
            0  & \text{otherwise}
            \end{cases}
        \]
    \end{itemize}
    in other words, $h_v^{(i)}$ is non-zero if and only if the projection of $F_v$ in the direction of $g_i$ has the most magnitude.
    \item  \textbf{Cheeger-type rounding:} 
    \begin{itemize}
        \item for each $h^{(i)}$, apply the Cheeger-type rounding algorithm of Lange et al. (Lemma \ref{lem:lange}) to obtain $y^{(i)}$
    \item output all $y\in \{y^{(i)}\}_{i\in [k]}$ satisfying $\mathcal{P}_p(y) = O\left(\sqrt{\log k }\right)\cdot \sqrt{\lambda_k(L_z)}$
    \end{itemize}
\end{enumerate}
\end{algorithmic}
\end{algorithm}

We summarize the changes in our algorithm from the algorithm of Louis et al. \cite{LouisRaghavendraTetaliVempala12}:
\begin{itemize}
    \item In the spectral embedding stage, we compute the embedding vectors using the eigenvectors of the \emph{rotated} Laplacian matrix, rather than those of the standard Laplacian matrix.
    \item In the Gaussian projection stage, Louis et al. \cite{LouisRaghavendraTetaliVempala12} project the embedding vectors onto independent real Gaussian vectors. For each $v\in V$, they set $h_v^{(i)}= \ip{F_v}{g_i}$ if and only if the projection $\ip{F_v}{g_i}$ has the largest \emph{value} among all $\ip{F_v}{g_j}, j\in [k]$. 
    
    In contrast, we project the embedding vectors onto independent \emph{complex} Gaussian vectors $g_1,\ldots, g_n$. For each $v\in V$, we set $h_v^{(i)}= \ip{F_v}{g_i}$ if and only if  the projection of $F_v$ onto the direction of $g_i$ has the largest \emph{magnitude}.   
    
    Our change follows the intuition we sketched in the introduction: for vertices from the same cluster $S$, their embedding vectors $\{F_v\}_{v\in S} \subseteq \C^k$ are expected to be  close to a common one-dimensional subspace $\operatorname{span}(g)$ for some $g\in \C^k, \|g\|=1$. Equivalently, the projection magnitude $|\ip{F_v}{g}|$ should be large for $v\in S$. The label information within the cluster $S$ is encoded by the phase of the projection $\ip{F_v}{g}$. In particular, if $S$ is an independent periodic component in $G$, then the vectors $\{F_v\}_{v\in S}$ live exactly in the one-dimensional subspace $\operatorname{span}(F_u)$ for any $u\in S$, differing only by complex phases. 
    \item In the Cheeger-type rounding stage, we use the rounding algorithm of Lange et al. \cite{lange2015frustration}, rather than the Fiedler's algorithm \cite{fiedler1973algebraic}.
\end{itemize}

Note that by definition the set of $\{h^{(i)}\}_{i\in [k]}$ have disjoint support, since for each $v\in V$ there exists at most one $i$ with the maximum $|\ip{F_v}{g_i}|$, and thus Lemma \ref{lem:lange} guarantees that the $y^{(i)}$'s are disjointly supported.

\section{Analysis of Algorithm \ref{algo:main}}\label{sec:analysis-main}
We will analyze the quality of $\{y^{(i)}\}_{i\in [k]}$ obtained from rounding $\{h^{(i)}\}_{i\in [k]}$ in the third stage. We will eventually show that each $ y^{(i)}$ has a small periodicity ratio with  constant probability over the complex Gaussians. For each $h^{(i)}, i\in [k]$, we apply the Lange et al. rounding algorithm (Lemma \ref{lem:lange}) to obtain a vector $y^{(i)}\in\{0,1,z^1, \ldots. z^{p-1}\}^n$ with performance

\[\mathcal{P}_p(y^{(i)}) \leq 2\cdot \frac{\sum_{u\to v} w(u,v)\cdot (|h^{(i)}_u|+|h^{(i)}_v|)\cdot |h^{(i)}_u-z\cdot h^{(i)}_v|}{\sum_v d_v|h^{(i)}_v|^2}\]

We will show the following bounds for the numerator and the denominator. 

\begin{lemma}\label{lem:all-bounds}
    For each $i\in [k]$, $h = h^{(i)}$ as constructed in Algorithm \ref{algo:main}, we have
    \begin{enumerate}
        \item \[\E_{g_1 \ldots, g_k \sim \cC\gau(0,1)^k}\left[\sum_{u\to v}   w(u,v)\cdot (|h_u|+|h_v|)\cdot |h_u-z\cdot h_v| \right] \leq 16\left(4c_1+\sqrt{2}\right)\cdot (\log k)^{3/2} \cdot \sqrt{\lambda_k(L_z)}\]
        where $c_1>0$ is an absolute constant (see Theorem \ref{fact:correlation}) 
        \item \[\frac{\log k}{2} \leq \E_{g_1 \ldots, g_k \sim \cC\gau(0,1)^k}\left[\sum_v d_v |h_v|^2\right] \leq 16 \log k\]
        \item \[\mathrm{Var}_{g_1 \ldots, g_k \sim \cC\gau(0,1)^k} \left[\sum_vd_v|h_v|^2\right] \leq 512(\log k)^2\]
    \end{enumerate}
\end{lemma}

These bounds are analogous to the corresponding bounds in \cite{LouisRaghavendraTetaliVempala12}, where Louis et al. worked with standard Laplacian matrix and real Gaussians. In below, we first use these bounds to complete the analysis of the algorithm, and then provide their proofs.

Since from now on it is clear that the expectation is over the randomness of the complex Gaussian variables, we will omit the subscript $g_1 \ldots, g_k \in \cC\gau(0,1)^k$ for convenience. 

\begin{proof}
    (of Theorem \ref{thm:main-algo}) 
    
    For each $i\in [k]$, $h=h^{(i)}$, let the random variables $X^{(i)}$ and $Y^{(i)}$ be
\[X^{(i)}=\sum_{u\to v}   w(u,v)\cdot (|h_u|+|h_v|)\cdot |h_u-z\cdot h_v|\] 
\[Y^{(i)}=  \sum_v d_v |h_v|^2 \]

By Lemma \ref{lem:lange}, we have $\cP_p(y^{(i)}) \leq 2\cdot X^{(i)}/Y^{(i)}$, so our task is to show that we can bound the latter with constant probability. By Lemma \ref{lem:all-bounds}, we have the expectation of $Y^{(i)}$ is $(\log k)/2\leq \E\left[Y^{(i)}\right] \leq 16\log k$, and the variance is at most $ 512 (\log k)^2$. By Chebyshev's inequality (Fact \ref{fact:cherbyshev}), using $t= \E[Y^{(i)}] / 2\sqrt{\mathrm{Var} \left[Y^{(i)}\right]}$:

\[\pr\left[Y^{(i)} \geq \frac{\E \left[Y^{(i)}\right]}{2}\right]\geq \frac{(\E\left[Y^{(i)}\right] )^2}{(\E\left[Y^{(i)} \right])^2 + 4\mathrm{Var}\left[Y^{(i)}]\right]} \geq \frac{1}{8193}=:c'  \]

Now apply the Markov inequality to our bound on the expectation of $X^{(i)}$, we have with probability at least $1-c'/2$,  $X^{(i)}$ is at most 
\[\frac{2}{c'}\cdot\E \left[X^{(i)}\right]  \leq  \frac{2}{c'}\cdot 16(4c_1+\sqrt{2})(\log k)^{3/2} \sqrt{\lambda_k(L_z)}\] 

By a union bound, we have with probability at least $c'/2$

\[\cP_p(y^{(i)}) \leq \frac{8}{c'}\cdot \frac{\E[X^{(i)}]}{\E[Y^{(i)}]} =  O\left(\sqrt{\lambda_k(L_z)\cdot \log k}\right)\]

Thus, the expected number of periodic components with this quality is at least $\Omega(k)$. The algorithm will produce with constant probability a collection of disjointly supported vectors $y^{(1)}, \ldots, y^{(\lceil ck\rceil)}\in \{0,z^0, z^1,\ldots, z^{p-1}\}^n$ for some $0<c<1$ such that 

\[\max_{i\in [\lceil ck\rceil]} \cP_p(y^{(i)}) =O\left(\sqrt{\lambda_k(L_z)\cdot \log k}\right) \]
\end{proof}

\hfill

The rest of this section is devoted to the proof of the bounds in Lemma \ref{lem:all-bounds}. We first need to present a few results that will be useful throughout this section.

\subsection{Preparation for the proof of Lemma \ref{lem:all-bounds}}
Recall the higher-order Rayleigh quotient of $D^{1/2}F$ with respect to $L_z$  is 
\[R_{L_z}(D^{1/2}F) :=  \frac{\sum_{u\to v} w(u,v)\cdot \|F_u -z\cdot  F_v\|^2}{\sum_v d_v \cdot \|F_v\|^2}\]
The following facts about spectral embedding are previously known for real symmetric graph Laplacian $L$ (see for example, \cite{spielman2019sagt, trevisan2017lecture}), here we extend them to the Hermitian, rotated graph Laplacian $L_z$. We defer the proofs to Appendix \ref{ap:prelim}.
\begin{proposition}\label{fact:embedding}
    For a graph $G=(V,E)$, $p\in \N$ and $z= e^{2\pi/p}$, let $L_z$ be its rotated Laplacian and let $F$ be the spectral embedding of $G$ using the first $k$ eigenvectors of $L_z$. We have
    \begin{enumerate}
        \item \[R_{L_z}(D^{1/2}F) \leq \lambda_k(L_z)\]
        \item \[\sum_{v\in V}d_v\|F_v\|^2 = k\]
        \item \[\sum_{u,v\in V}d_ud_v |\ip{F_u}{F_v}|^2 = k\]
    \end{enumerate}
\end{proposition}

For $F_u, F_v$, define the distance function 
\[\ang(F_u,F_v):= \min_{w\in\{1,z,\ldots, z^{p-1}\}} \|\bar{F_u} - w\cdot \bar{F}_v\|, \qquad \text{where }z=2\pi i/p\]

Recall that $\bar{x} = x/\|x\|$ is the normalized unit vector of $x$. Comparing to the \textit{radial projection distance} function $\text{dist}_r(F_u,F_v):= \|\bar{F}_u -\bar{F}_v \|$ as defined in \cite{lee2014multiway}, we extend their definition to  complex spaces, and incorporate a rotation $w$. We prove the following useful fact in the Appendix \ref{ap:prelim}:
\begin{lemma}\label{lem:angle-bound}
    For $z = e^{2\pi i/p}, p\in \N$, and $\{F_v\}_{v\in V}$ a collection of vectors in $\R^k$, then
    \[\sqrt{||F_u\|^2 + \|F_v\|^2} \cdot \ang(F_u,F_v) \leq 2\sqrt{2}\cdot \|F_u - z\cdot  F_v\|\]
\end{lemma}

We next present facts about Gaussian random variables and Gaussian projections.

\begin{fact}
    \cite{Gallager2008circularly} Let $f\in \C^k$ and $g\sim \cC\gau(0,I_k)$, then $\ip{f}{g}\sim \cC \gau(0,\|f\|^2)$
\end{fact}

In Appendix \ref{ap:prelim}, we prove the following facts about the expected maximum of $k$ independent complex Gaussian variables. 

\begin{proposition}\label{prop:gaussian-fact}
    Let $X_1, \ldots, X_k$ be $k$ independent standard complex Gaussian random variables. Let $Y = \max\{|X_i|\}$,
    \begin{enumerate}
        \item $\E[Y] \leq 4\sqrt{\log k}$, $\E[Y^2] \leq 16\log k$ and $\E[Y^4] \leq 64e(\log k)^2$
        \item $\E[Y] \geq  \sqrt{(\log k)/2 }$ and $\E[Y^2] \geq (\log k)/2$
    \end{enumerate}
\end{proposition}

Having introduced the facts above, we can show the following result about Gaussian projections of spectral embedding vectors. The proof is deferred to the end of this section (Subsection \ref{pf:prop-defer}).
\begin{proposition}\label{prop:condition-one}
    Let $\{F_v\}_{v\in V}$ be embedding vectors for $v\in V$, $g_1, \ldots, g_k \sim \cC\gau(0,I_k)$ be independent complex Gaussian vectors and $i\in [k]$. Let $h^{(i)}\in \C^n$ be defined by
    \[h^{(i)}_v:= \begin{cases}
        \ip{F_v}{g_i} & \text{if } |\ip{F_v}{g_i}|> |\ip{F_v}{g_j}|, \forall j\in [k]\setminus\{i\} \\
    0 &\text{otherwise}
    \end{cases}\]
     Consider $g=g_i$ and $h = h^{(i)}$, and  let $f\in \C^k$ be a fixed vector independent of $g$, we have 
    \begin{enumerate}
        \item \[\E\left[|\ip{f}{g}|^2 \cdot  \mathbb{I}[h_v \neq 0]\right]\leq \frac{16}{k}\cdot \|f\|^2 \cdot \log k\]
        and 
        \[\E\left[|\ip{f}{g}|^2 \cdot  \mathbb{I}[h_v = 0]\right] \leq 16\cdot \left(1-\frac{1}{k}\right)\cdot \|f\|^2 \cdot \log k\]
        \item for all pairs of vertices $u,v \in V$
    \[\pr\left[h_u\neq 0 \text{ and } h_v = 0\right] \leq c_1  \frac{\sqrt{\log k}}{k} \cdot \ang(F_u, F_v) \]
        for some absolute constant $c_1>0$.
    \end{enumerate}
\end{proposition}

\subsection{Analysis of the denominator}\label{sec:denominator}

\begin{lemma}
    (Item 2 of Lemma \ref{lem:all-bounds}.) For $i\in [k]$ and $h=h^{(i)}$ as in Algorithm \ref{algo:main}, we have 

    \[\frac{\log k}{2} \leq \E\left[\sum_v d_v |h_v|^2\right] \leq 16 \log k\]
\end{lemma}

\begin{proof}
    Since $g_1,\ldots, g_k$ are chosen independently at random, the random variables $\{\ip{F_v}{ g_j}\}_{j\in [k]}$ are  i.i.d.. Thus by symmetry,
    \[\pr\left[|\ip{F_v}{g_i}|>|\ip{F_v}{g_j}|, \forall j\neq i\right] =\frac{1}{k}\]
    and
    \[\E\left[\left|h_v^{(i)}\right|^2\right] = \frac{\E\left[\max_j |\ip{F_v}{g_j}|^2\right]}{k}\]
    By Proposition \ref{prop:gaussian-fact} we have 
    
    \[\frac{ \|F_v\|^2\log k}{2k}\leq \E\left[\left|h_v^{(i)}\right|^2\right] \leq \frac{16 \|F_v\|^2\log k}{k}\]
    Using item 2 of Proposition \ref{fact:embedding}, we have 
    \[\frac{\log k}{2} \leq \E\left[\sum_v d_v\left|h_v\right|^2\right] \leq 16\log k\]
\end{proof}

To bound the variance of the denominator, unlike in the real case where Louis et al. \cite{LouisRaghavendraTetaliVempala12} use real Hermite polynomials, here we need to extend the argument using complex Hermite polynomials (Lemma \ref{lem:correlation}). We first introduce necessary background on complex Hermite polynomials below.

\textbf{Complex Hermite polynomials and their properties}  

Denote the complex Gaussian space by $\cg = L^2(\C, \cC\gau(0,1))$, which is the Hilbert space of square-integrable functions with respect to the standard complex Gaussian measure $\cC\gau(0,1)$. The (normalized) complex Hermite polynomials $\{H_{m,l}\}_{m,l\in \Z_{\geq 0}}$ form an orthonormal basis for complex-valued functions over $\cg$ \cite{ali2014deformed, intissar2006spectral, van1990new}. In other words, we have $\E_{g\sim \cC\gau(0,1)} H_{m,l}(g)H_{m',l'}(g)^* = 1$ if $(m,l) = (m',l')$ otherwise $0$. Furthermore, the tensor product of these basis polynomials defined by

\[H_{m,l}(x_1,\ldots, x_k) = \prod_{i=1}^k H_{m_i,l_i 
}(x_i) \]

for $m,l \in \Z_{\geq 0}^k$ form an orthonormal basis for functions over $L^2(\C^k, \cC\gau(0,I_k))$. The degree of the polynomial $H_{m,l}(x)$ is $\text{deg}(H_{m,l}) = \sum_i (m_i+l_i)$. Consider the noise operator (\cite{ismail2016analytic}, 3.11) $T_\rho$ for $\rho \in \C, |\rho| \in [0,1]$, that takes a function $H:\C^k\to \C$ and returns the function $T_\rho H:\C^k\to \C$, defined as

\[(T_\rho H)(x)= \E_{Z\sim \cC \gau(0,1)}\left[H\left(\rho x+ \sqrt{1-|\rho|^2} Z\right)\right]\]

It is known that the complex Hermite polynomials form an eigenbasis for $T_\rho$ (\cite{ismail2016analytic}, 3.12) . Specifically, for every $m,l\in \Z \geq 0$,

\[T_\rho H_{m,l} = \rho^{\text{deg}(H_{m,l})} H_{m,l}\]

Thus, the following holds:
\begin{fact}\label{fact:hermite-eigen}
    For $\rho \in \C , |\rho|\in [0,1]$,
    \[\E_{g\sim \cC\gau(0,I_k)} \left[H_{m,l}(g)\left(T_{\rho^*} H_{m',l'}(g)\right)^*\right] =\begin{cases}
        \rho^{\text{deg}(H_{m,l})} & \text{if }(m,l) = (m',l')\\
        0 & \text{otherwise}
    \end{cases} \]
\end{fact}

Let $X$ and $Y$ be $\rho$-correlated Gaussians, i.e, $X\sim \cC\gau (0,1)$, $Y= \rho^* X + \sqrt{1-|\rho|^2}\cdot Z$ for an independent $Z\sim \cC\gau(0,1)$. Thus Fact. \ref{fact:hermite-eigen} implies:

\begin{corollary}\label{cor:noise}
    Let $(X_1, Y_1) ,\ldots, (X_k, Y_k) $ be independent pairs of  $\rho$-correlated complex Gaussians. We have for $m,m',l,l'\in \Z_{\geq 0}^k$, 

    \[\E \left[ H_{m,l}(X_1,\ldots, X_k)H_{m',l'}(Y_1,\ldots, Y_k)^*\right] =\begin{cases}
        \rho^{\text{deg}(H_{m,l})}  & \text{if }(m,l) = (m',l')\\
        0 & \text{otherwise}
    \end{cases} \]
\end{corollary}

The following is the main lemma we need for our analysis:

\begin{lemma}\label{lem:correlation}
    Let $u,v \in \C^k$ be unit vectors and $G_1,\ldots, G_k \sim \cC\gau(0,I_k)$ be i.i.d. complex Gaussian vectors. Let $X, Y\in \C^k$ be such that $X_j= \ip{u}{G_j}, Y_j = \ip{v}{G_j}$, so that for each $j\in [k]$, $X_j$ and $Y_j$ are $\rho$-correlated Gaussians for $\rho = \ip{u}{v}$. Fix $i\in [n]$, consider the function $f:\C^k \to \R$ defined by 
    \[f(x) =\begin{cases}
    |x_i|^2 & (\text{if } |x_i|^2> |x_j|^2 ~\forall j\neq i)\\
    0 & \text{otherwise}
    \end{cases}\]
    Then we have 
    \[\E \left[f(X)f(Y)\right] \leq \frac{256(\log k)^2}{k} \left(|\ip{u}{v}|^2 + \frac{1}{k}\right)\]
\end{lemma}
\begin{proof}
    By Proposition \ref{prop:gaussian-fact} we have $\E[f(X)]\leq  16(\log k)/k$ and $\E[f(X)^2]\leq 256(\log k)^2/k$. The function $f$ has the decomposition $f(x) = \sum_{m,l \in \Z_{\geq 0}^k} \widehat{f}_{m,l} H_{m,l}(x)$, and satisfies $\E[f(X)^2] = \sum_{m,l \in \Z_{\geq 0}^k} (\widehat{f}_{m,l})^2 $. Furthermore, $f$ is an even function since $f(x) = f(-x)$. We claim that the odd-degree coefficients of $f$ vanish, i.e, $\widehat{f}_{m,l}=0$ for $\text{deg}(H_{m,l})$ odd: let $H_{m,l}$ be an odd-degree basis polynomial, by definition (\cite{ismail2016analytic}) $H_{m,l}$ is an odd function if $\text{deg}(H_{m,l})$ is odd. Therefore, $f(x)H_{m,l}(x)^*$ is also an odd function as $f$ is even. Since the complex Gaussian measure $\cC\gau(0,I_k)$ is symmetric, the integral of any odd integrable function is $0$. Thus,
    \[\widehat{f}_{m,l} = \frac{\ip{f}{H_{m,l}}}{\|H_{m,l}\|^2} = \frac{\E_{X\sim \cC\gau(0,I_k)}f(X)H_{m,l}(X)^*}{\|H_{m,l}\|^2} = 0\]
    for $\text{deg}(H_{m,l})$ being odd. We have  
    \begin{align*}
        \E[f(X)f(Y)^*]
        &= \E\left[\left(\sum_{m,l \in \Z_{\geq 0}^k} \widehat{f}_{m,l} H_{m,l}(X)\right) \cdot \left(\sum_{m,l \in \Z_{\geq 0}^k} \widehat{f}_{m,l} H_{m,l}(Y)\right)^*\right]\\
        &= \sum_{\substack{m,l\in \Z_{\geq 0}^k}} \rho^{\text{deg}(H_{m,l})} |\widehat{f}_{m,l}|^2 \qquad \qquad \qquad \qquad  (\text{by Corollary }\ref{cor:noise}) \\
        &\leq \E[f]^2 +  \sum_{\substack{\text{deg}(H_{m,l})>1\\m,l\in \Z_{\geq 0}^k}} \rho^{\text{deg}(H_{m,l})} |\widehat{f}_{m,l}|^2  \qquad  \left(\widehat{f}_{0,0}=\E[f] \text{ and } \widehat{f}_{m,l} = 0 \text{ for } \text{deg}(H_{m.l}) =1\right)\\
        &\leq \E[f]^2 + |\rho|^2\cdot  \E[f^2] \qquad\qquad\qquad\qquad\qquad (*)\\      
        &\leq \frac{256(\log k)^2}{k^2} + |\ip{u}{v}|^2\cdot \frac{256(\log k)^2}{k}\\
        &\leq \frac{256(\log k)^2}{k}\left(|\ip{u}{v}|^2 + \frac{1}{k}\right)
    \end{align*}
    We explain the (*) step in more detail. We use the fact that for $\text{deg}(H_{m.l})\geq 2$, $|\rho^{\text{deg}(H_{m.l})}|\leq |\rho|^2$ since $0\leq |\rho|\leq 1$. Thus, since $f$ is a real-valued function,

\begin{align*}
    \sum_{\substack{\text{deg}(H_{m,l})>1\\m,l\in \Z_{\geq 0}^k}} \rho^{\text{deg}(H_{m,l})} |\widehat{f}_{m,l}|^2  &\leq |\rho|^2\cdot |\widehat{f}_{0,0}|^2 + \sum_{\substack{\text{deg}(H_{m,l})\geq1\\m,l\in \Z_{\geq 0}^k}} \rho^{\text{deg}(H_{m,l})} |\widehat{f}_{m,l}|^2  \\
    &\leq |\rho|^2\cdot |\widehat{f}_{0,0}|^2 + |\rho|^2\cdot \sum_{\substack{\text{deg}(H_{m,l})\geq1\\m,l\in \Z_{\geq 0}^k}}   |\widehat{f}_{m,l}|^2 \qquad (\text{for }\text{deg}(H_{m,l})=1, \widehat{f}_{m,l} = 0 )\\
    &= |\rho|^2\cdot \E\left[f^2\right] 
\end{align*}
\end{proof}

We are now ready to prove the bound on the variance.
\begin{lemma}
    (Item 3 of Lemma \ref{lem:all-bounds}.) For each $i\in[k]$, $h = h^{(i)}$ as constructed in Algorithm \ref{algo:main}, we have
    \[\mathrm{Var}\left[\sum_vd_v|h_v|^2\right] \leq 512 (\log k)^2\]
\end{lemma}
\begin{proof}
    \begin{align*}
        \mathrm{Var}\left[\sum_vd_v|h_v|^2\right]  &\leq \sum_{u,v}d_ud_v\|F_u\|^2\|F_v\|^2\cdot \E\left[\frac{|h_u|^2}{\|F_u\|^2}\cdot \frac{|h_v|^2}{\|F_v\|^2}\right]\\
        &\leq \sum_{u,v}d_ud_v\|F_u\|^2\|F_v\|^2\cdot \E\left[f\left(\ip{\bar{F}_u}{g_1},\ldots, \ip{\bar{F}_u}{g_k}\right)\cdot f\left(\ip{\bar{F}_v}{g_1},\ldots, \ip{\bar{F}_v}{g_k}\right)\right]\\
        &\leq \sum_{u,v}d_ud_v\|F_u\|^2\|F_v\|^2\cdot \frac{256(\log k)^2}{k}\left(|\ip{\bar{F}_u}{\bar{F}_v}|^2+\frac{1}{k}\right) \qquad  (\text{by Lemma }\ref{lem:correlation})\\
        &\leq \frac{256(\log k)^2}{k} \cdot \left(\sum_{u,v}d_ud_v\cdot  |\ip{F_u}{F_v}|^2+\frac{1}{k}\cdot \left(\sum_v d_v\|F_v\|^2\right)^2\right)\\
        &\leq 512\cdot (\log k)^2 \qquad\qquad\qquad\qquad\qquad\qquad\qquad\qquad\qquad (\text{by Proposition }\ref{fact:embedding})
    \end{align*}
    
\end{proof}

\subsection{Analysis of the numerator}
 For each edge $u\to v$, the expectation $\E\left[(|h_u|+|h_v|)|h_u-z\cdot h_v|\right]$ can be decomposed as 
\begin{align*}
    \E\left[\left(|h_u|+|h_v|\right)\cdot |h_u-z\cdot h_v|\right]  &=\E\left[(|h_u|+|h_v|)\cdot |h_u-z\cdot h_v| \cdot \I[h_u,h_v \neq 0]\right]  \\
    & \quad +\E\left[(|h_u|+|h_v|)\cdot|h_u-z\cdot h_v| \cdot \I[h_u\neq 0, h_v = 0]\right]  \\
    & \quad + \E\left[(|h_u|+|h_v|)\cdot |h_u-z\cdot h_v| \cdot \I[h_u= 0, h_v \neq 0]\right] 
\end{align*}

In the following we will bound each of the terms on the RHS separately.

\begin{proposition}\label{prop: condition-both-nonzero}
    \[\E\left[(|h_u|+|h_v|)\cdot |h_u-z\cdot h_v| \cdot  \I[h_u,h_v\neq 0]\right] \leq \frac{16\log k}{k}\cdot  (\|F_u\|+\|F_v\|)\cdot \|F_u - z\cdot F_v\|\]
\end{proposition}
\begin{proof}
    First, note that
    \begin{align*}
        &\E\left[(|h_u|+|h_v|)\cdot |h_u-z\cdot h_v| \cdot  \I[h_u,h_v\neq 0]\right] \\
        &= \E\left[|\ip{F_u}{g}|\cdot |\ip{F_u-z\cdot F_v}{g}| \cdot \I[h_u,h_v\neq 0] \right]  + \E\left[|\ip{F_v}{g}|\cdot |\ip{F_u-z\cdot F_v}{g}| \cdot  \I[h_u,h_v\neq 0] \right] \\
        &\quad  (\text{linearity of expectation}) \\ 
        &\leq \sqrt{\E\left[|\ip{F_u}{g}|^2 \cdot  \I[h_u,h_v\neq 0]\right]}\cdot \sqrt{\E\left[|\ip{F_u-z\cdot F_v}{g}|^2 \cdot \I[h_u,h_v\neq 0]\right]} \\
        &\quad + \sqrt{\E\left[|\ip{F_v}{g}|^2 \cdot  \I[h_u,h_v\neq 0] \right]}\cdot \sqrt{\E\left[ |\ip{F_u-z\cdot F_v}{g}|^2 \cdot  \I[h_u,h_v\neq 0] \right]}\\
        &\quad (\text{Cauchy-Schwarz})
    \end{align*}
    We have
    \begin{align*}
        \E\left[|\ip{F_u}{g}|^2 \cdot \I[h_u,h_v\neq 0]\right] &\leq \E\left[|\ip{F_u}{g}|^2 \cdot \I[h_u\neq 0]\right]\\
        &\leq \frac{16\|F_u\|^2 \log k}{k} && (\text{by item 1 of Proposition }\ref{prop:condition-one} \text{ and symmetry})
    \end{align*}

    Similarly, we have 
\[\E\left[|\ip{F_u}{g}|^2 \cdot \I[h_u,h_v\neq 0]\right] \leq \frac{16\|F_v\|^2 \log k}{k}\]

and 

\begin{align*}
    \E\left[|\ip{F_u-z\cdot F_v}{g}|^2 \cdot \I[h_u,h_v\neq 0] \right]&\leq \E\left[\max_j |\ip{F_u-z\cdot F_v}{g_j}|^2\cdot \I\left[h_v\neq 0\right]\right]\\
    &\leq \frac{16\|F_u-z\cdot F_v\|^2\log k}{k} \quad  (\text{by item 1 of Proposition }\ref{prop:condition-one} \text{ and symmetry})
\end{align*}

Putting together, 
    \[\E\left[(|h_u|+|h_v|)\cdot |h_u-z\cdot h_v| \cdot  \I[h_u,h_v\neq 0]\right] \leq \frac{16\log k}{k} \cdot (\|F_u\|+\|F_v\|)\cdot \|F_u - z\cdot F_v\|\]
\end{proof}

\begin{proposition}\label{prop:5.13}
    \[\E\left[(|h_u|+|h_v|)\cdot |h_u-z\cdot h_v| \cdot \I[h_u\neq 0, h_v = 0]\right] \leq \frac{16c_1}{k}\cdot \|F_u\|^2 \cdot  (\log k)^{3/2} \cdot \ang(F_u,F_v)\]
\end{proposition}

We defer the proof of Proposition \ref{prop:5.13} to Appendix \ref{ap:prop5.13}. Given the the bounds above, we prove the bound on the expectation of the numerator.
\begin{proof}
    (of the first item in Lemma \ref{lem:all-bounds}) adding up the bounds we have 
    \begin{align*}
        &\sum_{u\to v}\E\left[w(u,v)\cdot (|h_u|+|h_v|)\cdot |h_u-z\cdot h_v|\right] \\
        &\leq \sum_{u\to v}\frac{16\log k}{k}\cdot w(u,v)\cdot \Big((\|F_u\| + \|F_v\|)(\|F_u -z\cdot F_v\|) + c_1\cdot \sqrt{\log k}\cdot (\|F_u\|^2+\|F_v\|^2)\cdot \ang(F_u,F_v)\Big)\\
        &\leq \frac{16\log k}{k}\sum_{u\to v}w(u,v)\cdot \Big((\|F_u\| + \|F_v\|)(\|F_u -z\cdot F_v\|)  + 2\sqrt{2}c_1\cdot \sqrt{\log k}\cdot \sqrt{\|F_u\|^2+\|F_v\|^2}\cdot  \|F_u- z\cdot F_v\|\Big) \\&\qquad\qquad\qquad\qquad\qquad\qquad\qquad\qquad\qquad\qquad\qquad\qquad\qquad\qquad\qquad\qquad\qquad\qquad\qquad\quad (\text{by Lemma }\ref{lem:angle-bound})\\
        &\leq \frac{16\log k}{k}\cdot (2\sqrt{2}c_1\cdot \sqrt{\log k}+1)\sum_{u\to v}w(u,v)\cdot (\|F_u\| + \|F_v\|)\cdot \|F_u -z\cdot F_v\|\\
        &\leq \frac{16(2\sqrt{2}c_1+1)}{k}\cdot (\log k)^{3/2}\cdot \sqrt{\sum_{u\to v} w(u,v)\cdot (\|F_u\|+\|F_v\|)^2}\cdot \sqrt{\sum_{u\to v} w(u,v)\cdot \|F_u-z\cdot F_v\|^2} \quad (\text{Cauchy-Schwarz})\\
        &\leq \frac{16(2\sqrt{2}c_1+1)}{k}\cdot (\log k)^{3/2}\cdot \sqrt{\sum_{u\to v} 2w(u,v)\cdot (\|F_u\|^2 + \|F_v\|^2)}\cdot \sqrt{\lambda_k(L_z)\cdot \sum_v d_v\|F_v\|^2} \qquad (\text{by Proposition }\ref{fact:embedding})\\
        &\leq \frac{16(2\sqrt{2}c_1+1)}{k}\cdot(\log k)^{3/2}\cdot  \sqrt{2}\cdot \sqrt{\lambda_k(L_z)} \cdot \sqrt{\sum_{v} d_v\|F_v\|^2}\cdot \sqrt{\sum_v d_v\|F_v\|^2} \\
        &\leq 16(4c_1+\sqrt{2})\cdot (\log k)^{3/2}\cdot \sqrt{\lambda_k(L_z)} \qquad\qquad\qquad\qquad\qquad\qquad\qquad\qquad\qquad\qquad\qquad\quad (\text{by Proposition }\ref{fact:embedding})
    \end{align*}
\end{proof}

\subsection{Proof of Proposition \ref{prop:condition-one}} \label{pf:prop-defer}
We now prove Proposition \ref{prop:condition-one}. For the second item in Proposition \ref{prop:condition-one}, we will need the following theorem, which upper bounds the probability that  the largest magnitudes of two sequences of coordinate-wise correlated complex Gaussian random variables occur at different indices. As mentioned in the introduction, this is a complex analog of Theorem 4.1 in \cite{charikar2006near}. We present its proof in Appendix \ref{appendix}.
\begin{theorem}\label{fact:correlation} 
    Let $(X_1,Y_1), \ldots, (X_k, Y_k)$ be $k$ independent pairs of random variables, where $(X_i, Y_i)$ are $\rho_i$-correlated complex Gaussians.  If the average absolute covariance of $\{X_i\}$ and $\{Y_i\}$ is at least $1-\epsilon$, that is,

    \[\frac{\sum_i |\rho_i|}{k} \geq 1-\epsilon\]

    then 
    
    \[\pr[\arg\max_i |X_i| \neq \arg\max_i |Y_i|] \leq c_1 \sqrt{\epsilon\log k}\]
    for some absolute constant $c_1>0$.
\end{theorem}
\begin{proof}
    (Proof of Proposition \ref{prop:condition-one}.) For the first item, we have
    \begin{align*}
        \E\left[|\ip{f}{g}|^2 \cdot  \mathbb{I}[h_v \neq 0]\right] &\leq \E\left[\max_{l\in [k]} |\ip{f}{g_l}|^2 \cdot   \mathbb{I}\left[|\ip{F_v}{g_i}|^2> |\ip{F_v}{g_j}|^2, \forall j\neq i \right]\right] \\
         &\leq \frac{1}{k} \cdot \E\left[ \max_{l\in [k]} |\ip{f} {g_l}|^2\right] && (*)\\
         &= \frac{16}{k}\cdot \|f\|^2 \cdot \log k && (\text{by Proposition } \ref{prop:gaussian-fact})
    \end{align*}
    where the step of $(*)$ is by symmetry: given $\max_{l\in[k]} |\ip{f}{g_l}|^2$, $\arg \max_{j\in[k]} |\ip{F_v}{g_j}|^2$ is uniformly distributed in $[k]$.
    \begin{align*}
        \E\left[|\ip{f}{g}|^2 \cdot  \mathbb{I}[h_v = 0]\right] &\leq \E\left[\max_{l\in [k]} |\ip{f}{g_l}|^2 \cdot   \mathbb{I}\left[|\ip{F_v}{g_i}|^2 \leq |\ip{F_v}{g_j}|^2, \forall j\in [k]\backslash \{i\} \right]\right] \\
         &\leq \left(1-\frac{1}{k}\right) \cdot \E\left[ \max_{l\in [k]} |\ip{f}{g_l}|^2\right] \\
         &= 16\cdot \left(1-\frac{1}{k}\right)\cdot \|f\|^2 \cdot \log k && (\text{by Proposition } \ref{prop:gaussian-fact})
    \end{align*}

    Now we prove the second item. For $i\in [k]$, let $X_i = \ip{\bar{F}_u}{g_i}$ and $Y_i = \ip{\bar{F}_v}{g_i}$ so that $\rho_i = \E_{g_i}\left[X_iY_i^*\right]=\ip{\bar{F}_u}{\bar{F}_v} $. 
    We have 
    \begin{align*}
        |\rho_i| &= |\ip{\bar{F}_u}{\bar{F}_v}|\\
        &= \max_{|z| =1} \re(z\cdot  \ip{\bar{F}_u}{\bar{F}_v})\\
        &\geq \max_{w\in\{1,z,\ldots, z^{p-1}\}} \re(w\cdot  \ip{\bar{F}_u}{\bar{F}_v}) \\
        &= \max_{w\in\{1,z,\ldots, z^{p-1}\}} \left(1-  \frac{1}{2}\|\bar{F}_u - w\cdot \bar{F}_v\|^2\right) && (*)\\
        &=1-\min_{w\in\{1,z,\ldots, z^{p-1}\}}  \frac{1}{2}\|\bar{F_u} - w\cdot \bar{F}_v\|^2
    \end{align*}
    where $(*)$ holds because $\|\bar{F}_u - w\cdot \bar{F}_v\|^2 = \|\bar{F}_u\|^2 + \|\bar{F}_v\|^2 -w\ip{\bar{F}_u}{\bar{F}_v} - w^*\ip{\bar{F}_v}{\bar{F}_u}$, thus 

    \[\frac{\sum_i |\rho_i|}{k}\geq 1-\min_{w\in\{1,z,\ldots, z^{p-1}\}}  \frac{1}{2}\|\bar{F_u} - w\cdot \bar{F}_v\|^2\]
    Thus by Theorem \ref{fact:correlation},
    \begin{align*}
        \pr[h_u\neq 0, h_v= 0] &= \frac{1}{k}\sum_{j \in [k]} \pr[h_u^{(j)}\neq 0, h_v^{(j)} = 0] && (\text{by symmetry})\\
        &=  \frac{1}{k}\cdot \pr[\arg\max_j |X_j| \neq \arg\max_i |Y_j|]\\
        &\leq c_1 \frac{\sqrt{\log k}}{k} \cdot \min_{w\in\{1,z,\ldots, z^{p-1}\}} \|\bar{F_u} - w\cdot \bar{F}_v\|^2
    \end{align*}
    where the first equality holds because the $g_i$'s are i.i.d., thus by symmetry for each $g_j$, the probability $ \pr[h_u^{(j)}\neq 0, h_v^{(j)}= 0] $ is equal.
\end{proof}

\section{Acknowledgments}

We used GPT-5.5 Pro to assist with calculations in the proof of Proposition \ref{prop:5.13}, and with proof verification. We also used GPT-5.5 to improve the writing. 

Near the completion of this work, Jiyu learned of the passing of his Wai Po, his maternal grandmother. He would like to dedicate this work to her memory and thank her for the beautiful childhood memories she gave him.

\printbibliography

\appendix

\section{Cheeger-type rounding algorithm for periodicity}\label{sec:cheeger}

We present and analyze a Cheeger-type rounding algorithm for finding a complex vector $y\in \{0, z^0, z^1,\ldots, z^{p-1}\}^n$ with small periodicity ratio in a graph using the first eigenvector of its rotated Laplacian matrix. This algorithm is due to the work of Lange et al. \cite{lange2015frustration}, and can be viewed as a generalization of Trevisan's algorithm \cite{Trevisan2009MaxCut}, which is in turn a variant of Fiedler's algorithm \cite{fiedler1973algebraic}.

\begin{algorithm}[H]
\caption{Cheeger-type rounding for periodicity}
\label{algo:cheeger}
\begin{algorithmic} 
\State Input: graph $G=(V,E)$, a vector $x\in \C^n$, an integer $p\in \N$, $z = e^{2\pi i/ p}$.

\begin{enumerate}
    \item sort the value $x_v$ according to their square magnitude $|x_v|^2$
    \item for each $j\in [n]$, 
    \begin{enumerate}
        \item divide $V$ into $S_j:= \{v\in V\mid |x_v|^2\leq |x_j|^2\}$ and $V-S_j$. 
        \item sort $x_v = e^{i\alpha_v}\cdot |x_v|$ for $v\in V-S_j$ according to the angle $\alpha_v$ of $x_v$. 
        \item for each $u\in V-S_j$, divide the unit disk into angle intervals $[\alpha_u + 2\pi l/p, \alpha_u + 2\pi (l+1)/p) $ for $\ell \in \{0,\ldots, p-1\}$,
        set \[y^{(j, u)}_v = \begin{cases}
            z^l & \text{if } v\in V-S_j, \alpha_v \in [\alpha_u + 2\pi l/p, \alpha_u + 2\pi (l+1)/p)\\
            0 & \text{if } v\in S_j
        \end{cases}\]
        \item select $y^{(j)} \in \{y^{(j,u)}\}_{ u\in V-S_j}$ such that $\cP_p(y^{(j)})$ is minimized.
    \end{enumerate}
    \item output $y \in \{y^{(j)}\}_{j\in [n]}$ such that $\cP_p(y)$ is minimized.
\end{enumerate}
\end{algorithmic}
\end{algorithm}

The performance of the above algorithm can be analyzed using the following Lemma:
\begin{lemma}\label{lem:rounding}
    Let $x\in \C^n - \mathbf{0}$, such that $x_v = |x_v|\cdot e^{i\alpha_v}$. Consider the following randomized procedure for rounding $x$ into $y\in\{0,1,z^1, \ldots, z^{p-1}\}^n$:
    \begin{itemize}
        \item normalize $x$ so that $\max_i |x_i|^2 = 1$. 
        \item pick a uniformly random threshold $t^2$  from $[0,1]$, and a uniformly random $\alpha \in [0,2\pi)$
        \item set $y= \textnormal{Round}_{t,\alpha}(x_v)$ for all $v\in V$, where for all $s = |s|\cdot e^{i\theta}\in \C, |s|\leq 1$, we define 
        \[\textnormal{Round}_{t,\alpha}(s) = \begin{cases}
            z^l & \text{if } |s|^2\geq t^2, \theta \in [\alpha + 2\pi l/p, \alpha + 2\pi (l+1)/p)\\
            0 & \text{if }|s|^2< t^2
        \end{cases}\]
    \end{itemize}
    then 
    \[\frac{\E \left[\sum_{u\to v} w(u,v)\cdot |y_u-z\cdot y_z|\right]}{\E\left[\sum_vd_v\cdot |y_v|\right]} \leq 2\cdot \frac{\sum_{u\to v} w(u,v)\cdot |x_u-z\cdot x_v|\cdot (|x_u| + |x_v|)}{\sum_v d_v|x_v|^2} \]
    furthermore, the RHS can be upper bounded by $2\sqrt{2R_{L_z}(D^{1/2}x)}$.
\end{lemma}

\begin{proof}
    For the numerator, by linearity of expectation we will focus on one edge $u\to v$. By normalization we have $|x_v|^2, |x_u|^2\leq 1$.
    We use the following lemma due to Lange et al. without proof.
    \begin{lemma}\label{lem:lange-threshold}
        (see \cite{lange2015frustration}, Lemma 4.2.) For any two points $ s_1,s_2\in \C$ with $|s_1|,|s_2|\leq 1$, we have
        \[\E \left[|\textnormal{Round}_{t,\alpha}(s_1)- \textnormal{Round}_{t,\alpha}(s_2)|\right] \leq 2\cdot |s_1 - s_2|\cdot (|s_1| +|s_2|)\]
    \end{lemma}
    Therefore, taking $s_1 = x_u$ and $s_2 = z\cdot x_v$, and note that $\textnormal{Round}_{t,\alpha}(s_1)= y_u$ and $\textnormal{Round}_{t,\alpha}(s_2)= z\cdot \textnormal{Round}_{t,\alpha}(x_v) = z\cdot y_v$, by Lemma \ref{lem:lange-threshold} we have 
    \[\E \left[|y_u- z\cdot y_v|\right] \leq 2\cdot |x_u - z\cdot x_v|\cdot (|x_u| +|z\cdot x_v|) = 2\cdot |x_u - z\cdot x_v|\cdot (|x_u| +| x_v|)\]

    and by linearity of expectation, we have 
    \[\E\left[\sum_{u\to v} w(u,v)\cdot |y_u-z\cdot y_v|\right] \leq 2\cdot \sum_{u\to v} w(u,v)\cdot (|x_u| + |x_v|)\cdot |x_u - z\cdot x_v|\]

    For the denominator, we have 

    \begin{equation*}\tag{*}
        \E \left[\sum_v d_v\cdot |y_v|\right] = \sum_{v} d_v\cdot  \pr_{t^2\sim [0,1]}\left[|x_v|^2\geq t^2\right]  = \sum_v d_v|x_v|^2
    \end{equation*}

    Combining both bounds, we obtain the first upper bound. For the furthermore part, we further upper bound the expectation of the numerator by  
    \begin{align*}
        &2\cdot \sum_{u\to v} w(u,v)\cdot (|x_u| + |x_v|)\cdot |x_u - z\cdot x_v| \\
        &\leq 2\  \sqrt{\sum_{u\to v}w(u,v)\cdot (|x_u| + |x_v|)^2} \cdot \sqrt{\sum_{u\to v} w(u,v)\cdot |x_u - z\cdot x_v|^2} \quad  (\text{Cauchy-Schwarz})\\
        &\leq 2 \sqrt{2\cdot \sum_{u\to v}w(u,v)\cdot (|x_u|^2 + |x_v|^2)}\cdot \sqrt{\sum_{u\to v} w(u,v)\cdot |x_u - z\cdot x_v|^2}\\
        &\leq 2 \sqrt{2 \cdot \sum_v d_v|x_v|^2}\cdot  \sqrt{\sum_{u\to v} w(u,v)\cdot |x_u - z\cdot x_v|^2}\\
        &=2 \sqrt{2 \cdot \sum_v d_v|x_v|^2}\cdot  \sqrt{\sum_v d_v|x_v|^2 \cdot R_{L_z}(D^{1/2} x)}\\
        &= 2\sqrt{2\cdot R_{L_z}(D^{1/2}x)} \cdot \sum_v d_v|x_v|^2
    \end{align*}
    Combining the above with the expectation of the denominator (*) concludes the proof.

\end{proof}

We now apply Lemma \ref{lem:rounding} to prove Lemma \ref{lem:lange}

\begin{lemma}
    (Restatement of Lemma \ref{lem:lange}) For a digraph $G=(V,E)$, $p\in \N, z = e^{2\pi i/ p}$, given a complex vector $x\in \C^n$, 
    there exists a polynomial-time algorithm (Algorithm \ref{algo:cheeger}) that rounds $x$ into a vector $y\in \{0,z^0,z^1,\ldots, z^{p-1}\}^n$ with
    \[\cP_p(y) \leq 2\cdot \frac{\sum_{u\to v} w(u,v)\cdot |x_u-z\cdot x_v|\cdot (|x_u| + |x_v|)}{\sum_v d_v|x_v|^2} \]
    Moreover, $\supp(y)\subseteq \supp(x)$.
\end{lemma}
\begin{proof}
    (of Lemma \ref{lem:lange}) By the fact that if $X$ and $Y$ are random variables such that $\Pr[Y>0] = 1$, then with positive probability $\frac{X}{Y} \leq \frac{\E\left[X\right]}{\E \left[Y\right]}$, and we can conclude that there exists $t^2$, $\theta$ such that the corresponding  $y\in\{0,1,z^1, \ldots, z^{p-1}\}^n$ has periodicity ratio: 
    \[\cP_p(y) \leq 2\cdot \frac{\sum_{u\to v} w(u,v)\cdot |x_u-z\cdot x_v|\cdot (|x_u| + |x_v|)}{\sum_v d_v|x_v|^2} \]
    Moreover, we claim such a $y$ must have been considered in Algorithm \ref{algo:cheeger}. This is because $y$ is one of the level sets induced by a threshold value $t^2\in[0,1]$ and a threshold angle $\alpha\in [0,2\pi)$.  There are at most $n^2$ threshold sets since there are at most $n$ vertices (hence $n$ threshold values and $n$ threshold angles). Each of these level sets has been examined in Algorithm \ref{algo:cheeger}.
    
    Finally, observe that in the rounding procedure in Lemma \ref{lem:rounding}, $x_v= 0$ implies $y=0$ almost surely, thus we have $\supp(y)\subseteq \supp(x)$.
\end{proof}

Using the same argument and the bound of $2\sqrt{2R_{L_z}(D^{1/2}x)}$ in Lemma \ref{lem:rounding}, we can recover the 
upper bound in the Cheeger-type inequality for periodicity in Theorem \ref{thm:cheeger-unified}.

\section{Missing proofs for Section  \ref{sec:prelim}, \ref{sec:analysis-main}} \label{ap:prelim}

\begin{proposition}
    (Restatement of Proposition \ref{prop:Lz-property}) For $z=e^{2\pi i/p}$ for some $ p\in \mathbb{N}$, the normalized rotated Laplacian matrix is defined as $L_z = I-D^{-1/2}A_zD^{-1/2}$ has the following properties:
    \begin{enumerate}
        \item $L_z = D^{-1/2}B_z^*B_zD^{-1/2}$
        \item $L_z$ is Hermitian, positive semidefinite, and has eigenvalues $0\leq \lambda_1(L_z) \leq \cdots \leq \lambda_n(L_z)\leq 2$.
    \end{enumerate}
\end{proposition}

\begin{proof}
    For the first item, we claim that $B_z^*B_z = D-(zA)^*- zA$.  Let $B_v$ denote the $v$-th column of $B_z$. If there is an edge $u\to v$, by definition we have $B_z^*B_z(v,u) = \ip{B_v}{B_u} = -z\cdot w(u,v)$, and that $B_z^*B_z(u,v) = \ip{B_u}{B_v} = (-z^*)\cdot w(u,v)$. As for the diagonal entries, this edge contributes to $(-z)(-z^*)\cdot w(u,v) = w(u,v) $ for the in-degree of $v$, and $w(u,v)$ for the out-degree of $u$. Adding up all the edges, we get $B_z^*B_z(v,v) = D_{\text{in}}(v,v) + D_{\text{out}}(v,v)$.  

    For the second item, it is straightforward calculation to see that $L_z^* = L_z$. To see that $L_z$ is positive semidefinite, note that by the first item,
    \[\min_{x\in \C} x^*L_zx  = \min_{y = D^{-1/2}x \in \C} y^*B_z^*B_zy = \sum_{u\to v} w(u,v)\cdot |y_u - z\cdot y_v|^2 \geq 0\]
    To upper bound the eigenvalues, we use the Courant-Fischer theorem (Theorem \ref{thm:CF-eigen}). The Rayleigh quotient of any $x\in \C^n$ w.r.t $L_z$ is bounded by
    \begin{align*}
        \frac{\sum_{u\sim v} w(u,v) \cdot |x_u- z\cdot x_v|^2}{\sum_v d_v |x_v|^2} &\leq \frac{\sum_{u\sim v} w(u,v) \cdot |x_u- z\cdot x_v|^2}{\sum_v d_v |x_v|^2}\\
        &\leq \frac{\sum_{u\sim v} w(u,v)\cdot  ||x_u|+ |x_v||^2}{\sum_v d_v |x_v|^2}\\
        &\leq \frac{\sum_{u\sim v} w(u,v) 2(|x_u|^2 + |x_v|^2)}{\sum_v d_v |x_v|^2}\\
        &\leq \frac{2\sum_{v}d_v|x_v|^2}{\sum_v d_v |x_v|^2}\\
        &\leq 2
    \end{align*}
\end{proof}

\begin{proposition}
    (Restatement of Proposition \ref{fact:embedding})
    \begin{enumerate}
        \item \[R_{L_z}(D^{1/2}F) \leq \lambda_k(L_z)\]
        \item \[\sum_v d_v\|F_v\|^2 = k\]
        \item \[\sum_{u,v}d_ud_v |\ip{F_u}{F_v}|^2 = k\]
    \end{enumerate}
\end{proposition} 
\begin{proof}
    For the first item,
    \begin{align*}
        R_{L_z}(D^{1/2}F)&= \frac{\sum_{u\to v} w(u,v)\|F_v- z\cdot F_v\|^2}{ \sum_v d_v\|F_v\|^2} \\ 
        &= \frac{\sum_{u\to v}w(u,v)\sum_i |x_u^{(i)} - z\cdot  x_v^{(i)}|^2}{ \sum_v d_v\cdot \sum_i |x_v^{(i)}|^2} \\ 
        &= \frac{\sum_i \sum_{u\to v} w(u,v)|x_u^{(i)} -z\cdot  x_v^{{(i)}}|^2}{\sum_i \sum_v d_v|x^{(i)}_v|^2} \\ 
        &\leq \max_i \frac{\sum_{u\sim v} w(u,v)|x_u^{(i)} -z\cdot  x_v^{{(i)}}|^2}{ \sum_v d_v|x^{(i)}_v|^2} \\
        &\leq \max_i R_{L_z}(D^{1/2}x^{(i)}) \leq \lambda_k(L_z)
    \end{align*} 
    For the second item, notice that 

    \[\sum_v d_v \|F_v\|^2 = \sum_i \sum_v d_v |x^{(i)}_v|^2 = \sum_i \|D^{1/2}x^{(i)}\|^2 = k\]

    For the third item, 
    \begin{align*}
        \sum_{u,v}d_ud_v |\ip{F_u}{F_v}|^2 &= \sum_{u,v} |\ip{\sqrt{d_u}F_u}{\sqrt{d_v}{F_v}}|^2 \\ 
        &= \sum_{u,v} |D^{1/2}FF^*D^{1/2}(u,v)|^2 \\
        &= \|D^{1/2}FF^*D^{1/2}\|_{\text{Frob}}^2\\
        &=\text{tr}\left(D^{1/2}FF^*D FF^*D^{1/2}\right)\\
    \end{align*}
    by the definition of $F$, $W:=D^{1/2}F$ is the $n\times k$ matrix with the first $k$ orthonormal eigenvectors  of $L_z$ as its columns, so the above is equal to 
    \[\text{tr}\left(WW^*WW^*\right)= \text{tr}\left(WW^*\right) = \text{tr}\left(W^*W\right)=\text{tr}(I_k) = k\]
\end{proof}

\begin{lemma}
    (Restatement of Lemma \ref{lem:angle-bound}) For $z = e^{2\pi i/p}, p\in \N$,
    \[\sqrt{||F_u\|^2 + \|F_v\|^2} \cdot \ang(F_u,F_v) \leq 2\sqrt{2}\cdot \|F_u - z\cdot  F_v\|\]
\end{lemma}
\begin{proof}
    WLOG assume that $\|F_v\|\leq \|F_u\|$, and let $w$ minimize $\min_{w\in\{1,z,\ldots, z^{p-1}\}} \|\bar{F_u} - w\cdot \bar{F}_v\|$. We have 
    \begin{align*}
        \sqrt{||F_u\|^2 + \|F_v\|^2} \cdot \ang(F_u,F_v) &\leq
        \sqrt{2}\cdot \|F_u\|\cdot (\|\bar{F}_u - w\cdot \bar{F}_v\|)\\
        &\leq
        \sqrt{2}\cdot \|F_u\|\cdot (\|\bar{F}_u - z\cdot \bar{F}_v\|)\\
        &\leq \sqrt{2}\cdot \left\|F_u - \|F_u\| \cdot \frac{z\cdot F_v}{\|F_v\|}\right\| \\
        &\leq \sqrt{2}\cdot \left\|F_u - \|F_v\|\cdot \frac{z\cdot F_v}{\|F_v\|} + \|F_v\|\cdot \frac{z\cdot F_v}{\|F_v\|} -\|F_u\|\cdot \frac{z\cdot F_v}{\|F_v\|} \right\| \\
        &\leq \sqrt{2} \cdot \left(\left\|F_u -z\cdot  F_v\right\| + \left\|(\|F_u\|-\|F_v\|)\cdot  \frac{z\cdot F_v}{\|F_v\|}\right\| \right)\\
        &= \sqrt{2}\cdot \left(\|F_u -  z\cdot F_v\| + \left|\|F_u\| - \|F_v\|\right|\right) \\
        &\leq 2\sqrt{2}\cdot \|F_u- z\cdot F_v\|
    \end{align*}
\end{proof}

\begin{proposition}
    (Restatement of Proposition \ref{prop:gaussian-fact})
    Let $X_1, \ldots, X_k$ be $k$ independent standard complex Gaussian random variables. Let $Y = \max\{|X_i|\}$,
    \begin{enumerate}
        \item $\E[Y] \leq 4\sqrt{\log k}$, $\E[Y^2] \leq 16\log k$ and $\E[Y^4] \leq 64e(\log k)^2$
        \item $\E[Y] \geq  \sqrt{(\log k)/2 }$ and $\E[Y^2] \geq (\log k)/2$
    \end{enumerate}
\end{proposition}

The following fact about the expected maximum of $k$ independent Gaussian variables is given in \cite{LouisRaghavendraTetaliVempala12} for the real case: 

\begin{fact}\label{prop:gaussian-fact-LRTV}
    (\cite{LouisRaghavendraTetaliVempala12}, Fact 3.8) Let $X_1, \ldots, X_k$ be $k$ independent standard real Gaussian random variables. Let $Y = \max\{X_i\}$,
    \begin{enumerate}
        \item $\E[Y] \leq 2\sqrt{\log k}$, $\E[Y^2] \leq 4\log k$ and $\E[Y^4] \leq 4e(\log k)^2$
        \item $\E[Y] \geq \sqrt{\log k }$ and $\E[Y^2] \geq \log k$
    \end{enumerate}
\end{fact}

Here we prove the complex case by a reduction to the real case.

\begin{proof} (of Proposition \ref{prop:gaussian-fact})
    Note that for any $i$, 
    \[|X_i|^4 = (\re(X_i)^2 + \im(X_i)^2)^2 \leq  2\re(X_i)^4 + 2\im(X_i)^4 \leq 4\max \{\re(X_i)^4, \im(X_i)^4\}\]
    We then have
    \begin{align*}
        \E[Y^4]&\leq \E\left[\max_i |X_i|^4 \right] \\
        &\leq \E\left[ 4\cdot \max \left\{\re(X_1)^4, \im(X_1)^4, \ldots, \re(X_k)^4, \im(X_k)^4  \right\} \right]\\
        &\leq 4\left( 4e \cdot (\log 2k)^2\right) && (\text{by Fact \ref{prop:gaussian-fact-LRTV}})\\
        &\leq 64e \cdot (\log k)^2 &&(\log 2k\leq \log 2 +\log k \leq  2\log k)
    \end{align*}
    Apply Jensen's inequality we obtain
    \[\E[Y^2]  = \E[(Y^4)^\frac{1}{2}]\leq  \sqrt{\E[Y^4]} \leq 16\log k\]
    and 
    \[\E[Y] \leq \sqrt{\E[Y^2]}\leq 4\sqrt{\log k}\]

    For the second item, we have 
    \[\E[Y] = \E \left[\max_i |X_i|\right] \geq \E \left[\max_i |\re(X_i)|\right] \geq \frac{1}{\sqrt{2}} \sqrt{\log k}\]
    and by Jensen's inequality
    \[\E[Y^2] \geq \E[Y] ^2 \geq \frac{1}{2}\log k\]
    
\end{proof}

\section{Proof of Proposition \ref{prop:5.13}}\label{ap:prop5.13}

We state a few useful facts. We first recall Harris’s Inequality, which gives a criterion for proving non-positive correlation between monotone functions of independent random variables
\begin{theorem}
    \label{thm:harris}
    (Harris's Inequality) Let $f:\R\to \R$ be non-decreasing function and $g:\R\to \R$ be non-increasing functions. For a real random variable $X$, we have
    \[\E[f(X)g(X)] \leq \E[f(X)]\cdot \E[g(X)]\]
\end{theorem} 
We will need to reason about the magnitude of complex Gaussians. The magnitude of a complex Gaussian $X$ is equal to the magnitude of a real Gaussian vector where the components are the real and imaginary part of $X$. We summarize its distribution in the following:
\begin{fact}
    \label{fact:rayleigh-rice}
    Let $X = (Y_1, Y_2)$ be a real Gaussian vector where $Y_i\sim \gau(\mu_i,\sigma^2)$ for $i=1,2$ are independent Gaussians, we have the magnitude of the vector, i.e, $|X|$, has the following distribution:
    \begin{itemize}
    \item if $\mu_1= \mu_2 = 0$, $|X|$ has a Rayleigh distribution \cite{rayleigh-distribution}, with density 
    \[f_{|X|}(r) = \frac{r}{\sigma^2}\cdot e^{-r^2/2\sigma^2}, \qquad r\geq 0\]
    and CDF
    \[\pr\left[|X|\leq r\right] = 1-e^{-r^2/2\sigma^2}, \qquad r\geq 0\]

    \item if $\mu_1\geq 0, \mu_2= 0$, then $|X|$ has a Rice distribution \cite{rice-distribution}, with density 
    \[f_{|X|}(r) = \frac{r}{\sigma^2}\cdot e^{-\frac{r^2 +\mu_1^2}{2\sigma^2}}\cdot I_0\left(\frac{r\mu_1}{2\sigma^2}\right)\]
    where 
    \[
    I_0(z)= \int_0^{2\pi}e^{z\cos \theta}\rmd \theta
    \]
    is called the modified Bessel function of the first kind of order $0$ \cite{modified-bessel}.

    \end{itemize} 
\end{fact}
The following properties of modified Bessel function will be used:

\begin{fact}\label{fact:bessel}
    (\cite{modified-bessel}) 
    For $m\in \N$, 
    \[I_m:= \int_0^{2\pi}e^{z\cos \theta}\cos (m\theta)\rmd \theta\] 
    is the modified Bessel function of the first kind of order $m$. In particular,
    \begin{enumerate}
        \item $I_0'(z) = I_1(z)$
        \item $I_1'(z) = I_0(z) -\frac{1}{z}\cdot I_1(z)$
    \end{enumerate}
\end{fact}

We now begin our proof of Proposition \ref{prop:5.13}.

\begin{proposition}
    (Restatement of Proposition \ref{prop:5.13})
    \begin{equation}
    \tag{*}\label{eq:5.13}\E\left[(|h_u|+|h_v|)\cdot |h_u-z\cdot h_v| \cdot \I[h_u\neq 0, h_v = 0]\right] \leq \frac{16c_1}{k}\cdot \|F_u\|^2\cdot  (\log k)^{3/2} \cdot \ang(F_u,F_v)
    \end{equation}
\end{proposition}

\begin{proof} 
    We first show that Proposition \ref{prop:5.13} can be reduced to proving the monotonicity of the probability $\pr\left[h_v\neq0 ~\Big\vert~ |h_u|^2  = x\right]$.

    \begin{claim}\label{claim:logic1}
        If the probability $\pr\left[h_v\neq0 ~\Big\vert~ |h_u|^2  = x\right]$ is non-decreasing for $x>0$, then Proposition \ref{prop:5.13} holds.
    \end{claim}
    \begin{proof}
        The LHS of the inequality $(\ref{eq:5.13})$ is  
    \begin{align*}
    \E\left[(|h_u|+|h_v|)\cdot |h_u-z\cdot h_v| \cdot \I[h_u\neq 0, h_v = 0]\right] 
        &= \E\left[|h_u|^2\cdot \I[h_v = 0, h_u\neq 0]\right]\\
        &\leq \E\left[\max_l |\ip{F_u}{g_l}|^2\cdot \I[h_u\neq 0, h_v = 0]\right]
    \end{align*}
    on the other hand, by Proposition \ref{prop:gaussian-fact} and Proposition \ref{prop:condition-one}, for the RHS we have 
    \[\E\left[\max_l |\ip{F_u}{g_l}|^2\right]\cdot \pr\left[h_u\neq 0, h_v=0 \right] \leq \frac{16c_1}{k}\cdot \|F_u\|^2\cdot  (\log k)^{3/2} \cdot \ang(F_u,F_v)\]

    so it suffices to prove that 

    \[\E\left[\max_l |\ip{F_u}{g_l}|^2\cdot \I[h_u\neq 0, h_v = 0]\right] \leq \E\left[\max_l |\ip{F_u}{g_l}|^2\right]  \cdot \E\left[\I\left[h_u\neq 0, h_v = 0\right]\right]  \]
    
    That is, we would like to show that the random variables $X:= \max_l |\ip{F_u}{g_l}|^2$ and $Y:= \I[h_u\neq 0, h_v = 0]$ are non-positively correlated. By Harris's Inequality (Theorem \ref{thm:harris}) with $f(X) =X$ it suffices to show that the function
    
    \[g(X)=\E\left[Y\mid X=x\right] = \pr\left[h_u\neq 0, h_v=0 ~\Big\vert~ \max_l |\ip{F_u}{g_l}|^2 = x\right]\]
    is a non-increasing function of $x$. Note that we can consider $x>0$ since  $\max_l |\ip{F_u}{g_l}|^2 = x> 0$ almost surely.  Furthermore, by symmetry we have 
    \begin{align*}
        &\pr\left[h_u \neq 0, h_v=0 ~\Big\vert~ \max_l |\ip{F_u}{g_l}|^2 = x\right] \\
        &= \pr\left[h_u\neq  0 ~\Big\vert~ \max_l |\ip{F_u}{g_l}|^2 = x\right] \cdot \pr\left[h_v =  0 ~\Big\vert~ h_u\neq 0, \max_l |\ip{F_u}{g_l}|^2 = x\right]\\
        &=\frac{1}{k}\cdot \pr\left[h_v=0 ~\Big\vert~ |h_u|^2  = x\right]
    \end{align*}
    Thus it reduces to proving $\pr\left[h_v=0 ~\Big\vert~ |h_u|^2  = x\right]$ is non-increasing for $x>0$, which is equivalent to proving $\pr\left[h_v\neq0 ~\Big\vert~ |h_u|^2  = x\right]$ is non-decreasing for $x>0$.
    \end{proof}
     We next define a few random variables and events that allows us to analyze the monotonicity of $\pr\left[h_v\neq0 ~\Big\vert~ |h_u|^2  = x\right]$.  
    
    Consider the random variables $\{G_i = |\ip{F_v}{g_i}|\}_{i\in [k]}$. W.l.o.g. we assume  $h = h^{(i)} = h^{(1)} $, that is, we fix $i= 1$. We decompose $F_v = \alpha F_u + \beta F_v'$ where $F_v' \perp F_u$. Then for each $j\in [k]$ we let $A_j:=\ip{\alpha F_u}{g_j}$ and $B_j:=\ip{\beta F_v'}{g_j}$, and we have $G_j = |A_j +  B_j| $. Consider the following probability:\begin{equation}\label{eq:prob-monotone}
        p(x) := \pr\left[G_1>G_l, \forall l\geq 2 ~\Big\vert~ |A_1|= x \text{ and }  \forall j\geq 2, |A_j|\leq x\right]
    \end{equation}
    By the definition of $h$, it is straightforward to see it has the same monotonicity as $\pr\left[h_v\neq0 ~\Big\vert~ |h_u|^2  = x\right]$:
    \begin{claim}\label{claim:logic2}
        $p(x)$ is non-decreasing in $x>0$ if and only if $\pr\left[h_v\neq0 ~\Big\vert~ |h_u|^2  = x\right]$ is non-decreasing in $x>0$.
    \end{claim}
    
    Thus it suffices to prove $p'(x)\geq 0$ for $x>0$, which will be our goal for the rest of the proof.   
    
    We next introduce more notations for convenience. 
    Let $E_x$ denote the event $\{|A_1| = x \text{ and }  \forall j\geq 2, |A_j|\leq x\}$, and let $G_1\vert E_x$ denote the  random variable $G_1$ conditioned on $E_x$.  Let $g(t,x)=\pr\left[G_l\leq t ~\Big\vert~ |A_l|\leq x\right]$ for a fixed $l\geq 2$.
    \begin{proposition}\label{prop:second-term} 
        We have  
        \begin{enumerate}
            \item 
            \[p(x) = \int_0^\infty f_{G_1 \lvert E_x}(t)\cdot g(t,x)^{k-1}\rmd t\]
            \item 
            \[p'(x) = (k-1)\int_0^{\infty}g(t,x)^{(k-2)}f_{G_1 \lvert E_x}(t)\cdot \left[\partial_x g(t,x) + \frac{I_1\left(2xt/\|\beta F_v'\|^2\right)}{I_0\left(2xt/\|\beta F_v'\|^2\right)}\cdot \partial_t g(t,x)\right] \rmd t\]
        \end{enumerate}
    For convenience, we will denote 

    \[\rho_x(t):= \frac{I_1\left(2xt/\|\beta F_v'\|^2\right)}{I_0\left(2xt/\|\beta F_v'\|^2\right)}\] 
    \end{proposition}
    \begin{proof}
        For the first item, observe that $A_1, \ldots, A_k$ and $B_1, \ldots, B_k$ are independent because $g_1,\ldots, g_k$ are independent complex Gaussians and the projection of a complex Gaussian onto orthogonal directions $\alpha F_u$ and $\beta F_v'$ are independent complex Gaussians. Therefore we have 
        \[p(x) = \int_0^\infty f_{G_1 \lvert E_x}(t)\cdot \prod_{l\geq 2}\pr\left[G_l\leq t ~\Big\vert~ |A_l|\leq x\right]\rmd t =\int_0^\infty f_{G_1 \lvert E_x}(t)\cdot g(t,x)^{k-1}\rmd t
    \]
        where the last equality holds because for $l\geq 2$,  the random variables $G_l$ conditioned on $|A_l|\leq x$ are identically distributed. 
        
        We next prove the second item. For any $l\in [k], t\geq 0$, by independence of $A_l$ and $B_l$ and rotational symmetry in the complex plane of $B_l$, we have   $\pr\left[
    |A_l+B_l| \leq t  ~\Big\vert~ |A_l| = x\right]=\pr\left[
    |x+B_l|\leq t\right]$. Thus $G_1\vert E_x$ is identically distributed to $|x+B_1|$, which is the magnitude of a Gaussian vector and has a Rice distribution (second item in Fact \ref{fact:rayleigh-rice} , with parameters $\mu_1=x$ and $\sigma^2= \|\beta F_v'\|^2/2$) . We now take the differentiation of $p(x)$ using the product rule:

\begin{align*}
    p'(x) &= \left(\int_0^\infty f_{G_1 \lvert E_x}(t)\cdot g(t,x)^{k-1}\rmd t\right)'\\
    &= (k-1)\int_0^\infty g(t,x)^{k-2}\cdot  \partial_xg(t,x) \cdot f_{G_1 \lvert E_x}(t)\rmd t+\int_0^\infty g(t,x)^{k-1}\cdot \partial_x f_{G_1 \lvert E_x}(t)\rmd t    
\end{align*}
    We show the following identity for the second term in $p'(x)$:
    \begin{claim}\label{claim:second-item}
        $\int_0^\infty g(t,x)^{k-1}\cdot \partial_x f_{G_1 \lvert E_x}(t)\rmd t=(k-1)\int_0^\infty g(t,x)^{k-2}\cdot f_{G_1\vert E_x}(t)\cdot \rho_x(t) \cdot \partial_t g(t,x)\rmd t$ 
    \end{claim}
    and the second item follows by plugging the above identity into $p'(x)$ above. To prove Claim \ref{claim:second-item}, we make the following observation:
    \begin{observation}\label{obs:identity}
        $\partial_x f_{G_1\vert E_x}(t)= \partial_t\left[-\rho_x(t)\cdot f_{G_1\vert E_x}(t)\right]$. In particular, for any $l\in [k]$, $\partial_x \pr\left[|x+B_l|\leq t\right] = -f_{|x+B_l|}(t)\cdot \rho_x(t)$
    \end{observation}
    This follows by a direction calculation: recall that by Fact \ref{fact:rayleigh-rice}. 
    \[f_{G_1\vert E_x}(r) = \frac{2r}{\|\beta F_v'\|^2}\cdot \exp\left(-\frac{r^2 + x^2}{\|\beta F_v'\|^2}\right)\cdot I_0\left(\frac{rx}{\|\beta F_v'\|^2}\right)\]
    Differentiate with respect to $x$:
    \begin{align*}
        \partial_x f_{G_1\vert E_x}(t) &= \partial_x\partial_t \pr\left[|x+B_1|\leq t\right]
        =\partial_t\partial_x \int_0^t f_{G_1 \lvert E_x}(r)\rmd r =\partial_t\int_0^t\partial_x f_{G_1 \lvert E_x}(r)\rmd r\\
        &=\partial_t\int_0^t \frac{2r}{\|\beta F_v'\|^2}\cdot \exp\left(-\frac{r^2+x^2}{\|\beta F_v'\|^2}\right)\cdot \left[-\frac{2x}{\|\beta F_v'\|^2}\cdot I_0\left(\frac{2xr}{\|\beta F_v'\|^2}\right)+ \frac{2r}{\|\beta F_v'\|^2}\cdot I_1\left(\frac{2xr}{\|\beta F_v'\|^2}\right)\right] \rmd r\\
        &\qquad (\text{using item 1 of Fact \ref{fact:bessel}})\\
    &= \partial_t\int_0^t \partial_r\left[-\frac{2r}{\|\beta F_v'\|^2}\cdot \exp\left(-\frac{r^2+x^2}{\|\beta F_v'\|^2}\right)\cdot I_1\left(\frac{2xr}{\|\beta F_v'\|^2}\right)\right] \rmd r \\
    &\qquad (\text{expand and using item 2 of Fact \ref{fact:bessel}})\\
    &= \partial_t\left[-\frac{2r}{\|\beta F_v'\|^2}\cdot \exp\left(-\frac{r^2+x^2}{\|\beta F_v'\|^2}\right)\cdot I_1\left(\frac{2xr}{\|\beta F_v'\|^2}\right)\right]_{r=0}^{r=t}\\
    &= \partial_t\left[-\frac{2r}{\|\beta F_v'\|^2}\cdot \exp\left(-\frac{t^2+x^2}{\|\beta F_v'\|^2}\right)\cdot I_1\left(\frac{2xt}{\|\beta F_v'\|^2}\right)\right]\\
    &=\partial_t \left[-f_{G_1\vert E_x}(t) \cdot \rho_x(t)\right]
    \end{align*}
    \begin{proof}
        (of Claim \ref{claim:second-item}.) Using integration by part,
        \begin{align*}
    \int_0^{\infty} g(t,x)^{k-1}\partial_x f_{G_1 \lvert E_x}(t) \rmd t &= \int_0^{\infty} g(t,x)^{k-1}\partial_t \left[-\rho_x(t)f_{G_1 \lvert E_x}(t)\right] \rmd t \\
    &=\left[g(t,x)^{(k-1)}\cdot (-\rho_x(t))\cdot f_{G_1 \lvert E_x}(t)\right]_{t=0}^{t=\infty}\\
    &\quad -(k-1) \int_0^\infty g(t,x)^{(k-2)}\cdot\left[\partial_t~g(t,x)\right] \cdot \left[-\rho_x(t) \cdot f_{G_1 \lvert E_x}(t)\right] \rmd t\\
    &= (k-1)\int_0^\infty g(t,x)^{(k-2)}\left[\partial_t g{(t,x)}\right]\cdot \rho_x(t) \cdot f_{G_1 \lvert E_x}(t)\rmd t
\end{align*}
where in the last equality, we used the fact that $g(t,x)= 0$ at $t=0$, and $f_{G_1\vert E_x}(t) = 0$ at $t=\infty$, and that $0<\rho_x(t)<1$ for $x,t\geq 0$. Therefore, we have $g(t,x)^{(k-1)}\cdot (-\rho_x(t))\cdot f_{G_1 \lvert E_x}(t)$ vanishes at $t=0$ and $t=\infty$.
    \end{proof}
    
    \end{proof}

\begin{claim}\label{claim:logic3}
    For $x>0$, $p'(x)\geq 0 $ if  $\partial_x g(t,x) +  \rho_x(t)\cdot \partial_t g(t,x) \geq 0$
\end{claim}

\begin{proof}
    By Proposition \ref{prop:second-term}, 
    \[p'(x) = (k-1)\int_0^{\infty}g(t,x)^{(k-2)}f_{G_1 \lvert E_x}(t)\cdot \left[\partial_x g(t,x) + \rho_x(t)\cdot \partial_t g(t,x)\right] \rmd t
\]

both $g(t,x)$ and $f_{G_1\vert E_x}(t)$ are non-negative for $x>0 $ and $t\geq 0$, so it suffices to show the term $\partial_x g(t,x) + \rho_x(t)\cdot \partial_t g(t,x) \geq 0$.
\end{proof}

We next state the monotonicity of two quantities and use it to prove  $\partial_x g(t,x) +  \rho_x(t)\cdot \partial_t g(t,x) \geq 0$.

\begin{proposition}\label{prop:monotonicity}
    \hfill
    \begin{enumerate}
        \item For any $t\geq 0$, and $0<x_1\leq x_2$, 
        \[\rho_{x_1}(t)\leq \rho_{x_2}(t)\]
        \item For any $l\in [k]$ and $0<x_1\leq x_2$,
        \[\frac{ f_{|A_l|}(x_2) }{\pr\left[|A_l|\leq x_2\right]} \leq \frac{f_{|A_l|}(x_1)}{\pr\left[|A_l|\leq x_1\right]}\]
    \end{enumerate}
\end{proposition}

We defer the proof to the end. 

\begin{proposition}\label{prop:logic4}
    For $x>0$, $\partial_x g(t,x) +  \rho_x(t)\cdot \partial_t g(t,x) \geq 0$
\end{proposition}
\begin{proof}
    For any $l\in [k]$, $|A_l|$ is the magnitude of a two-dimensional Gaussian with $\mu_1=\mu_2=0$ and $\sigma^2 = \|\alpha F_v\|^2/2$ (item 1 of Fact \ref{fact:rayleigh-rice}), which has a Rayleigh distribution. Denote its density by $f_{|A_l|}(r)$ for $r\geq 0$, then 
    
    \[g(t,x)
    = \frac{\int_0^x \pr\left[G_l \leq t ~\Big\vert~ |A_l| =r\right] \cdot f_{|A_l|}(r) \rmd r}{\pr\left[|A_l|\leq x\right]}
    =\frac{\int_0^x \pr\left[|r+B_l|\leq t \right] \cdot f_{|A_l|} (r)\rmd r}{\pr\left[|A_l|\leq x\right]}
\]
and a direct calculation gives

\begin{align*}
    \partial_x g(t,x) &= \frac{ f_{|A_l|}(x) }{\pr\left[|A_l|\leq x\right]} \cdot \left( \pr\left[|x+B_l| \leq t  \right] -g(t,x)\right)\\
    &= \frac{ f_{|A_l|}(x) }{\pr\left[|A_l|\leq x\right]^2} \cdot \left(\pr\left[|A_l|\leq x\right]\cdot \pr\left[|x+B_l|\leq t\right]-\pr\left[|x+B_l|\leq t \land |A_l|\leq x|\right]\right)\\
    &=\frac{ f_{|A_l|}(x) }{\pr\left[|A_l|\leq x\right]^2} \cdot \left(\int_0^x f_{|A_l|}(r)\cdot \left(\pr\left[|x+B_l|\leq t\right] - \pr\left[|r+B_l|\leq t\right]\right)\rmd r\right)\\
    &=\frac{ f_{|A_l|}(x) }{\pr\left[|A_l|\leq x\right]^2}\cdot \left(\int_0^x f_{|A_l|}(r)\cdot \left(\int_r^x \left(-\rho_s(t)\cdot f_{|s+B_l|}(t)\right)\rmd s\right) \rmd r\right) \qquad (\text{by Observation \ref{obs:identity}})\\
    &= -\frac{ f_{|A_l|}(x) }{\pr\left[|A_l|\leq x\right]^2}\cdot \left(\int_0^x \rho_s(t)\cdot f_{|s+B_l|}(t)\cdot \left(\int_0^sf_{|A_l|}(r) \rmd r\right)\rmd s\right)  \qquad (\text{by Fubini})\\
    &= -\frac{ f_{|A_l|}(x) }{\pr\left[|A_l|\leq x\right]^2}\cdot \left(\int_0^x \rho_s(t)\cdot f_{|s+B_l|}(t)\cdot \pr\left[|A_l|\leq s\right]\rmd s\right)
\end{align*}

and to prove our proposition, it suffices to prove $\rho_x(t)\cdot \partial_t g(t,x) \geq -\partial_x g(t,x)$. Indeed, we have

\begin{align*}
    \rho_x(t)\cdot \partial_t g(t,x) 
    &= \rho_x(t)\cdot \frac{\partial_t\int_0^x f_{|A_l|}(s) \cdot \pr\left[|s+B_l|\leq t\right] \rmd s}{\pr\left[|A_l|\leq x\right]}\\
    &=\rho_x(t)\cdot \frac{\int_0^x f_{|A_l|}(s) \cdot f_{|s+B_l|(t) \rmd s}}{\pr\left[|A_l|\leq x\right]}\\
    &\geq \frac{1}{\pr\left[|A_l|\leq x\right]}\cdot \left(\int_0^x f_{|A_l|}(s) \cdot \rho_s(t) \cdot f_{|s+B_l|}(t) \rmd s\right) \qquad (\text{by item 1 of Proposition \ref{prop:monotonicity}})\\
    &= \frac{1}{\pr\left[|A_l|\leq x\right]}\cdot \left(\int_0^x \frac{f_{|A_l|}(s)}{\pr\left[|A_l|\leq s\right]}\cdot \pr\left[|A_l|\leq s\right] \cdot \rho_s(t)\cdot  f_{|s+B_l|}(t) \rmd s\right)\\
    &\geq \frac{ f_{|A_l|}(x) }{\pr\left[|A_l|\leq x\right]^2}\cdot \left(\int_0^x  \pr\left[|A_l|\leq s\right]\cdot \rho_s(t) \cdot f_{|s+B_l|}(t) \rmd s\right) \qquad (\text{by item 2 of Proposition \ref{prop:monotonicity}})\\
    &= -\partial_x g(t,x)
\end{align*}
as desired.
\end{proof}

Thus combining Claim \ref{claim:logic1}, Claim \ref{claim:logic2}, Claim \ref{claim:logic3}, and Proposition \ref{prop:logic4}, we have proved Proposition \ref{prop:5.13}. It remains to prove Proposition \ref{prop:monotonicity}.

\begin{proof}
    (of Proposition \ref{prop:monotonicity}.) For the first item, it suffices to show that $\rho(z):=I_1(z)/I_0(z)$ is non-decreasing for $z\geq 0$. Since for any fixed $t\geq 0$ and $x>0$, $z(x,t)= 2xt/\|\beta F_v'\|^2$ is non-decreasing in $x$. We take the differentiation of $\rho(z)$:
    \begin{align*}
        \rho(z)' &= \left(\frac{I_1(z)}{I_0(z)}\right)' = \left(\frac{I_0(z)'}{I_0(z)}\right)  =  \frac{I_0''(z)I_0(z)-I_0'(z)^2}{I_0(z)^2}
    \end{align*}
    where the second equality follows by $I_1(z) = I_0(z)'$ (item 1 of Fact \ref{fact:bessel}). so it suffices to show the numerator is non-negative. We can compute
    \[
I_0'(z)
=
\frac{1}{2\pi}
\int_0^{2\pi}
\cos\theta\ e^{z\cos\theta}\,d\theta
\qquad \text{ and } \qquad I_0''(z)
=
\frac{1}{2\pi}
\int_0^{2\pi}
\cos^2\theta\, e^{z\cos\theta}\,d\theta.
\]
    and so the numerator is 

    \begin{align*}
        I_0''(z)I_0(z)-I_0'(z)^2&= \left(
\int_0^{2\pi}
\cos^2\theta\, e^{z\cos\theta}\,d\theta
\right)
\left(
\int_0^{2\pi}
e^{z\cos\theta}\,d\theta
\right)
-
\left(
\int_0^{2\pi}
\cos\theta\, e^{z\cos\theta}\,d\theta
\right)^2\\
&\geq 0  \qquad (\text{Cauchy-Schwarz})
    \end{align*}
where the last inequality used the Cauchy-Schwarz Inequality:

\[
\left(\int_a^b f(\theta)g(\theta)\,d\theta\right)^2
\le
\left(\int_a^b f(\theta)^2\,d\theta\right)
\left(\int_a^b g(\theta)^2\,d\theta\right)
\]
with 
$f(\theta):=\cos\theta\, e^{(z\cos\theta)/2}$
and $g(\theta):=e^{(z\cos\theta)/2}.
$

\hfill

We next prove the second item. Recall that $|A_l|$ is the magnitude of a two-dimensional Gaussian with $\mu_1=\mu_2=0$ and $\sigma^2 = \|\alpha F_v\|^2/2$, which has a Rayleigh distribution (item 1 of Fact \ref{fact:rayleigh-rice}). We have

\begin{align*}
    \left(\frac{ f_{|A_l|}(x) }{\pr\left[|A_l|\leq x\right]}\right)' 
    &=  \left(\frac{2x}{\sigma^2}\cdot \frac{\exp\left(-\frac{x^2}{\sigma^2}\right)
}{
1-\exp\left(-\frac{x^2}{\sigma^2}\right)
}\right)'\\
    &= \frac{2}{\sigma^2}\cdot \frac{
e^{x^2/\sigma^2}\left(1-\frac{2x^2}{\sigma^2}\right)-1
}{
(e^{x^2/\sigma^2}-1)^2
}
\end{align*}

It suffices to show that the numerator $e^{x^2/\sigma^2}\left(1-\frac{2x^2}{\sigma^2}\right)-1\leq 0$. Equivalently, let $y:= x^2/\sigma^2\geq 0 $, we want to show
\[1-2y\leq e^{-y}\] for $y\geq 0$. This is true because  for $y\geq0$, $1-2y\leq 1-y \leq e^{-y}$.

\end{proof}
\end{proof}

\section{Proof of Theorem \ref{fact:correlation}.}\label{appendix}
Recall that two complex Gaussians $X$ and $Y$ are $\rho$-correlated if $X\sim \cC\gau (0,1)$, $Y= \rho^* X + \sqrt{1-|\rho|^2}\cdot Z$ for an independent $Z\sim \cC\gau(0,1)$. As a corollary, the covariance $\mathrm{cov}(X,Y)=\E\left[XY^*\right]= \rho$.
\begin{theorem}

    (Restatement of Theorem \ref{fact:correlation})
    Let $(X_1,Y_1), \ldots, (X_k, Y_k)$ be $k$  independent pairs of complex random variables, where $(X_i, Y_i)$ are $\rho_i $-correlated complex Gaussians.  If the average absolute covariance of $\{X_i\}$ and $\{Y_i\}$ is at least $1-\epsilon$, that is,

    \[\frac{\sum_i |\rho_i|}{k} \geq 1-\epsilon\]

    then 
    
    \[\pr[\arg\max_i |X_i| \neq \arg\max_i |Y_i|] \leq c_1 \sqrt{\epsilon\log k}\]
    for some absolute constant $c_1>0$.
\end{theorem}

A similar result for the real Gaussian case was established in Theorem 4.1 of \cite{charikar2006near}, and our proof is inspired by theirs. We first prove several useful results.

\begin{lemma}\label{lem:C.2}
    Let $X_1,X_2$ and $Y_1,Y_2$ be standard complex Gaussian variables where 
    \begin{enumerate}

        \item  $(X_1, Y_1)$ is independent of $(X_2, Y_2)$
        \item $\forall i\in \{1,2\}, (X_i,Y_i)$ is $\rho_i$-correlated and  $|\rho_i|\geq 1-\epsilon$.
    \end{enumerate}
    Then there exists complex Gaussians $Z_1,Z_2$ independent of $X_1,X_2$, and each with variance at most $2\epsilon$ such that with probability $1$ we have
    \[|Y_1|-|Y_2|\geq (1-\epsilon)|X_1|- |X_2| -|Z_1| - |Z_2|\]
\end{lemma}
\begin{proof}
    We can decompose $Y_1$ as $Y_1 = \rho_1^* X_1 + Z_1$ where $Z_1$ is an independent complex Gaussian. Similarly $Y_2 = \rho_2^* X_2 + Z_2$. Since $\var[Y_1] = 1 = \E[|\rho_1^* X_1 + Z_1|^2] = \E[|\rho_1|^2\cdot |X_1|^2] + \E[|Z_1|^2] = |\rho_1|^2 + \var[Z_1]$, we have $\var[Z_1] = 1-|\rho_1|^2\leq 2\epsilon$. Similarly $\var[Z_2] = 1-|\rho_2|^2\leq 2\epsilon$. We have
    \begin{align*}
        |Y_1| - |Y_2| &= |\rho_1^* X_1 + Z_1|- |\rho_2^* X_2 + Z_2|\\
        &\geq |\rho_1|\cdot |X_1| - |Z_1| - |\rho_2|\cdot |X_2| - |Z_2|\\
        &\geq (1-\epsilon)\cdot |X_1| - |X_2| - |Z_1| -|Z_2| && (\text{by assumption } 1-\epsilon\leq |\rho_1|,|\rho_2| \leq 1)
    \end{align*}
\end{proof}

Now, for a sequence of $k$ \textit{i.i.d.} standard complex Gaussian variables $X_1,\ldots. X_k$ and $t\geq 0$, we let the events $E_t, E$ be
\[E_t:= \{|X_1|= t \text{ and }|X_i| \leq t \text{ for all } 2\leq i\leq k\}\]
\[E := \{|X_1| \geq |X_i|, \forall i\in [k]\} = \cup_{t\geq 0}E_t\]

\begin{lemma}\label{lem:APdifference-condition}
    Let $(X_1,Y_1), \ldots, (X_k, Y_k)$ be $k$ independent pairs of random variables such that  $X_1,Y_1$ are $\rho_1$-correlated complex Gaussians and $X_2, Y_2$ are $\rho_2$-correlated complex Gaussians, with $|\rho_1|,|\rho_2|\geq 1-\epsilon$. Then for $\epsilon, t\geq 0$, we have:
    
    \hfill
    
    if $\epsilon t^2 \leq 1$,
    \[\pr\left[|Y_1|\leq |Y_2| ~\big|~ E_t\right]  \leq \frac{e^{-t^2}}{1-e^{-t^2}}\cdot  O(t\sqrt{\epsilon})\]
    if $\epsilon t^2 \geq1$ and $t>1$, 
    \[\pr\left[|Y_1|\leq |Y_2| ~\big|~ E_t\right]  \leq O(\sqrt{\epsilon})\]
\end{lemma}
\begin{proof}
    By Lemma \ref{lem:C.2}, we have
    \begin{align*}
        \pr\left[|Y_1|\leq |Y_2| ~\big|~ E_t\right] &\leq \pr\left[|X_2|+|Z_1|+|Z_2| \geq (1-\epsilon)\cdot |X_1| ~\big\vert~ E_t\right] \\
        &\leq \pr\left[|X_2| \geq (1-\epsilon)t - \left(|Z_1|+|Z_2|\right) ~\big \vert~ E_t\right] && (|X_1|=t \text{ under }E_t)\\
    \end{align*}
    $|X_2|$ is a Rayleigh distribution (see Fact \ref{fact:rayleigh-rice}), and for every fixed value $z$ of $|Z_1| + |Z_2|$, we have 
    \begin{align*}
        \pr\left[|X_2| \geq (1-\epsilon)t - z ~\big \vert~ E_t\right] &= \frac{\pr\left[(1-\epsilon)t - z \leq |X_2| \leq t\right]}{\pr\left[|X_2|\leq t\right]}\\
        &\leq \frac{e^{-(t-\epsilon t - z)^2} - e^{-t^2}}{1-e^{-t^2}}\\
        &\leq \frac{e^{-t^2}(e^{2t(\epsilon t+ z)- (\epsilon t + z)^2} - 1)}{1-e^{-t^2}}\\
        &\leq \frac{e^{-t^2}(e^{2t(\epsilon t+ z)} - 1)}{1-e^{-t^2}}
    \end{align*}

    Since $Z_1$ and $Z_2$ are independent of $X_2$,  we take the expectation of the above bound over $|Z_1| + |Z_2|$, and get

    \begin{align*}
        \pr\left[|Y_1|\leq |Y_2| ~\big|~ E_t\right]&\leq \E_{z\sim |Z_1|+|Z_2|}\left[\pr\left[|X_2| \geq (1-\epsilon)t - z ~\big \vert~ E_t\right]\right]\\
        &\leq \frac{e^{-t^2}}{1-e^{-t^2}}\cdot  \E_{|Z_1|+|Z_2|}\left[e^{2t(\epsilon t + |Z_1| + |Z_2|)} - 1\right]\\
        &\leq \frac{e^{-t^2}}{1-e^{-t^2}}\cdot\left(\left(e^{2\epsilon t^2}-1\right )\cdot \E_{|Z_1|+|Z_2|}\left[e^{2t(|Z_1| + |Z_2|)}\right] + \E_{|Z_1|+|Z_2|}\left[e^{2t(|Z_1| + |Z_2|)} -1\right]\right)
    \end{align*}
    When $\epsilon t^2 \leq 1$, the following claim holds:
    \begin{claim}
        If $\epsilon t^2 \leq 1$,  then $\E_{|Z_1|}\left[e^{2t |Z_1|}-1\right]\leq O(t\sqrt{\epsilon})$
    \end{claim}
    \begin{proof}
        Let $G_1 \sim \cC\gau(0,1)$ we have $ \var\left[Z_1\right] \leq 2\epsilon \cdot \var\left[G_1\right]$ (Lemma \ref{lem:C.2}) and $\E\left[\left(2t\cdot |Z_1|\right)^m\right]\leq (4t\sqrt{\epsilon})^m\cdot \E\left[|G_1|^m\right]$ for every $m\geq 0$. Since $\epsilon t^2\leq 1$, let $\lambda$ be an absolute constant such that $\lambda\geq 4t\sqrt{\epsilon}$. Using Taylor expansion:

    \begin{align*}
    \E\left[e^{2t\cdot |Z_1|}-1\right]&=\E\left[\sum_{m=1}^\infty\frac{(2t)^{m}|Z_1|^m}{m!}\right] \\
    &\leq  \E\left[\sum_{m=1}^\infty\frac{(4t\sqrt{\epsilon})^{m}|G_1|^m}{m!}\right]\\
    &\leq \frac{4t\sqrt{\epsilon}}{\lambda}\cdot \E\left[\sum_{m=1}^\infty\frac{\lambda^{m}|G_1|^m}{m!}\right] \\
    &\leq \frac{4t\sqrt{\epsilon}}{\lambda}\cdot \E\left[e^{\lambda |G_1|}\right]\\
    &\leq O(t\sqrt{\epsilon}) \qquad \qquad \qquad \qquad(\text{by the fact }\E\left[e^{\lambda |G_1|}\right] <\infty\text{ for every }\lambda>0)
    \end{align*}
    \end{proof}

    Using the above claim, by independence of $Z_1$ and $Z_2$, and the fact that that when $\epsilon t^2 \leq 1$, $e^{2\epsilon t^2}-1 = O(\epsilon t^2) = O(t\sqrt{\epsilon})$, we have the upper bound:
    \begin{align*}
        \pr\left[|Y_1|\leq |Y_2| ~\big|~ E_t\right]  &\leq \frac{e^{-t^2}}{1-e^{-t^2}}\cdot \left(O(t\sqrt{\epsilon})\cdot  \left(1+O(t\sqrt{\epsilon}) \right)^2+ \left(1+O(t\sqrt{\epsilon})\right)^2-1\right)\\
        &\leq \frac{e^{-t^2}}{1-e^{-t^2}}\cdot O(t\sqrt{\epsilon})
    \end{align*}
 
    When $\epsilon t^2 \geq 1$ and  $t>1$. If $\epsilon > 1/4$, $\pr\left[|X_2| \geq (1-\epsilon)t - \left(|Z_1|+|Z_2|\right) ~\big \vert~ E_t\right] \leq 1\leq O(\sqrt{\epsilon})$ is trivial. Thus we assume $\epsilon \leq 1/4$. Now we divide into two cases:

    If $|Z_1|+|Z_2| \leq t/4$, then
    \begin{align*}
        \pr\left[|X_2| \geq (1-\epsilon)t - \left(|Z_1|+|Z_2|\right) ~\big \vert~ E_t\right] &\leq \pr\left[|X_2| \geq t/2 ~\big\vert~ |X_2|\leq t\right]\\
        &\leq O(e^{-t^2/4})  \qquad \qquad (\text{by Fact \ref{fact:rayleigh-rice}})\\
        &\leq O(e^{-1/4\epsilon}) \qquad \qquad (\epsilon t^2\geq 1) \\
        &\leq O(\sqrt{\epsilon}) \qquad \qquad \qquad (\epsilon\leq 1/4)
    \end{align*}
    
    By Markov's Inequality,  
    \[\pr\left[|Z_1| + |Z_2| > t/4\right]\leq \frac{4\E\left[|Z_1| + |Z_2|\right]}{t} \leq O(\sqrt{\epsilon}/t)\leq O(\sqrt{\epsilon})\]
    where we used the fact that $\E\left[|Z_1|\right] = \E\left[|Z_2|\right]= \sqrt{\E\left[|Z_1|^2\right]} \leq O(\sqrt{\epsilon})$ and $t>1$.

\end{proof}

\begin{lemma}\label{lem:APtail}
    \[\pr\left[|X_1|\geq 2\sqrt{\log k} ~\big|~ E\right]\leq \frac{1}{k}\]
\end{lemma}
\begin{proof}
    The probability of $E$ is $1/k$ since for every $i$, the event $|X_i| = \max_{j\in [k]} |X_j|$ has equal probability. Therefore,
    \[\pr\left(|X_1|\geq 2\sqrt{\log k} ~\big|~ E\right)\leq \frac{\pr\left[|X_1|\geq 2\sqrt{\log k}\right]}{1/k}  \leq \frac{1/k^2}{1/k} \leq \frac{1}{k}\]
\end{proof}

\begin{lemma}\label{lem:APtwovar}
    Let $(X_1,Y_1), \ldots, (X_k, Y_k)$ be $k$  independent pairs of random variables, such that $X_1,Y_1$ are $\rho_1$-correlated complex Gaussians and $X_2, Y_2$ are $\rho_2$-correlated complex Gaussians, with $|\rho_1|,|\rho_2|\geq 1-\epsilon$, where $\epsilon< 1/(4\log k)$. Then
    \[\pr\left[|Y_1|\leq|Y_2| ~\big|~ E\right] = O\left(\frac{\sqrt{\epsilon\log k}}{k}\right)\]
\end{lemma}
\begin{proof}
    We decompose the probability as
    \begin{align*}
        \pr\left[|Y_1|\leq|Y_2| ~\big|~ E\right] &=\pr\left[|Y_1|\leq|Y_2| \land |X_1|\leq 2\sqrt{\log k}  ~\big|~ E\right]\\
        &+ \pr\left[|Y_1|\leq|Y_2| \land |X_1|\geq 2\sqrt{\log k}  ~\big|~ E\right]
    \end{align*}
    First consider the case $|X_1|\leq 2\sqrt{\log k}$ so that $\epsilon|X_1|^2\leq 1$ and let $\rmd F_{|X_1|~\vert E}$ be the probability density of $|X_1|$ conditioned on $E$, we have
    \begin{align*}
        \pr\left[|Y_1|\leq|Y_2| \land |X_1|\leq 2\sqrt{\log k}  ~\big|~ E\right] &= \int_0^{2\sqrt{\log k}} \pr\left[|Y_1|\leq|Y_2|   ~\big|~ E_t\right] \rmd F_{|X_1|~\vert E}(t)\\
        &\leq O(\sqrt{\epsilon})\cdot \int_0^{2\sqrt{\log k}} t\cdot \frac{e^{-t^2}}{1-e^{-t^2}}\rmd F_{|X_1|~\vert E}(t) && (\text{by Lemma \ref{lem:APdifference-condition}})
    \end{align*}  
    Since $X_i$'s are i.i.d. and continuous, the event $E$ that $|X_1|= \max_i |X_i|$ occurs with probability $1/k$, and is independent of the event that $|X_i|\leq t$ for all $i$, thus
    \begin{align*}
        \pr\left[\forall i \leq k, |X_i|\leq t ~\big\vert~ E \right] &= F(t)^k
    \end{align*}
    where $F(t) = \pr\left[X_i\leq t\right]$. Take the differentiation, we have $\rmd F_{|X_1|~\vert E}(t) = \int_0^t f_{|X_1|}(t)\cdot k\cdot F(t)^{k-1} \rmd t$, where $f_{|X_1|}(t) = \rmd F(t)$ is the density of $|X_1|$. Plug this back into the probability above:
    \begin{align*}
        &\pr\left[|Y_1|\leq|Y_2| \land |X_1|\leq 2\sqrt{\log k}  ~\big|~ E\right]\\
        &\leq O(\sqrt{\epsilon})\cdot \int_0^{2\sqrt{\log k}} t\cdot \frac{e^{-t^2}}{1-e^{-t^2}}\cdot f_{|X_1|}(t)\cdot k\cdot F(t)^{k-1} \rmd t\\
        &\leq O(\sqrt{\epsilon})\cdot k\cdot  \int_0^{\infty} 2t^2\cdot (e^{-t^2 })^2\cdot (1-e^{-t^2})^{k-2} \rmd t \qquad \qquad (|X_1| \text{ has a Rayleigh distribution, see Fact }\ref{fact:rayleigh-rice})\\
        &\leq O(\sqrt{\epsilon})\cdot k\cdot  \int_0^1 \sqrt{\log (1/x)}\cdot x\cdot (1-x)^{k-2} \rmd x\\
        &\leq  O(\sqrt{\epsilon})\cdot k\left( \int_0^{1/k} \sqrt{\log (1/x)}\cdot x\cdot (1-x)^{k-2} \rmd x + \int_{1/k}^{1} \sqrt{\log (1/x)}\cdot x\cdot (1-x)^{k-2} \rmd x\right)\\
    \end{align*}
    where the first term in the bracket is at most \[\int_0^{1/k} \sqrt{\log (1/x)}\cdot  \rmd x \leq O(\sqrt{\log k}/k)\] 
    since  $x\leq 1/k$ and $(1-x)^{k-2}\leq 1$. The second term is at most 
    \[\int_{1/k}^{1} \sqrt{\log (1/x)}\cdot x\cdot (1-x)^{k-2} \rmd x \leq \int_0^1 \sqrt{\log k}\cdot x(1-x)^{k-2}\rmd x = O\left(\frac{\sqrt{\log k}}{k(k-1)}\right) = O(\sqrt{\log k}/k)\]
    So we have
    \[\pr\left[|Y_1|\leq|Y_2| \land |X_1|\leq 2\sqrt{\log k}  ~\big|~ E\right] \leq O\left(\frac{\sqrt{\epsilon \log k}}{k}\right)\]

    Now consider the case $|X_1|\geq 2\sqrt{\log k}\geq 1$ and $\epsilon|X_1|^2\geq 1$. Let $\mathcal{A}:= E\cap \{|X_1|\geq 2\sqrt{\log k}\}$. By Lemma \ref{lem:APdifference-condition}, 
    \begin{align*}
        \pr\left[|Y_1|\leq|Y_2| ~\Big|~   \mathcal{A}\right]  &= \int_{2\sqrt{\log k}}^{\infty} \pr\left[|Y_1|\leq|Y_2|   ~\big|~ E_t\right] \rmd F_{|X_1|~\mid \mathcal{A}}(t)\\
        &\leq \int_{2\sqrt{\log k}}^{\infty} O(\sqrt{\epsilon})  \ \rmd F_{|X_1|~\mid \mathcal{A}}(t)\\
        &=O(\sqrt{\epsilon})
    \end{align*}
    and by Lemma \ref{lem:APtail},
    \[\pr\left(|X_1|\geq 2\sqrt{\log k} ~\big|~ E\right)\leq \frac{1}{k}\]
    Multiplying the above we have $\pr\left[|Y_1|\leq|Y_2| \land |X_1|\geq 2\sqrt{\log k}  ~\big|~ E\right] \leq O(\sqrt{\epsilon}/k)$. Finally, adding the bounds from both cases  we have
    \[\pr\left[|Y_1|\leq|Y_2| ~\big|~ E\right] = O\left(\frac{\sqrt{\epsilon\log k}}{k}\right)\]
\end{proof}
\begin{lemma}\label{lem:APfinal}
    Let $(X_1,Y_1), \ldots, (X_k, Y_k)$ be $k$ independent pairs of random variables. For all $i$. $X_i$ and $Y_i$ are $\rho_i$-correlated with $|\rho_i|\geq 1-\epsilon_i$. Suppose that
    \begin{enumerate}
        \item $\sum_i \epsilon_i / k = \epsilon$
        \item $\forall i, \epsilon_i \leq 1/(4\log k)$
    \end{enumerate}
    Then 
    \[\pr\left[\arg\max_i |X_i| = \arg\max_j |Y_j|  \right] \geq  1- O(\sqrt{\epsilon \log k})\]
\end{lemma}
\begin{proof}
    By Lemma \ref{lem:APtwovar} we have 
    \[\pr\left[|Y_1|\leq  |Y_2| ~\big|~ |X_1| = \max_{i\in [k]} |X_i|\right] = O\left(\frac{\sqrt{\max \{\epsilon_1, \epsilon_2 \} \log k}}{k}\right)\]
    By a union bound,
    \begin{align*}
        \pr\left[|Y_1| \leq \max_{j\neq 1} |Y_j| ~\big|~ |X_1| = \max_{i\in [k]} |X_i|\right] &= O\left(\frac{\sqrt{\log k}}{k}\cdot  \sum_{j=2}^k \sqrt{\max\{\epsilon_1,\epsilon_j\}}\right)\\
        &\leq O\left(\frac{\sqrt{\log k}}{k} \cdot  \sum_{i=1}^k (\sqrt{\epsilon_1} + \sqrt{\epsilon_j})\right)\\
        &\leq O\left(\sqrt{\log k} \cdot (\sqrt{\epsilon_1} + \sqrt{\epsilon })\right) && (\text{Jensen's inequality})
    \end{align*}
    Now, the probability $|X_i| = \max_{j\in [k]} |X_j|$ is equal for each $i$, which is $1/k$. Multiplying this and take another union bound, we have
    \begin{align*}
        \pr\left[\arg\max_i |X_i| \neq \arg\max_j |Y_j|  \right] &\leq   \sum_{i\in [k]} O\left(\sqrt{\log k} \cdot (\sqrt{\epsilon_i} + \sqrt{\epsilon })\right) \cdot \frac{1}{k}\\
        &\leq O\left(\sqrt{\log k}\cdot (\sqrt{\epsilon} + \sqrt{\epsilon })\right) && (\text{Jensen's inequality})\\
        &= O\left(\sqrt{\epsilon\log k } \right)
    \end{align*}
\end{proof}

Finally, we are ready to prove Theorem \ref{fact:correlation}.
\begin{proof}
    (of Theorem \ref{fact:correlation}.)  Let $\epsilon_i = 1-|\rho_i|$ and $\epsilon = (\sum_i \epsilon_i)/k$. Assume $\epsilon < 1/(4\log k)$, otherwise the theorem is trivial.

    Consider the set $I= \{i:  \epsilon_i <  1/(4\log k)\}$. Since $\epsilon< 1/(4\log k)$, the set is not empty. Apply Lemma \ref{lem:APfinal} to $\{X_i\}_{i\in I}$ and $\{Y_i\}_{i\in I}$, we have that
    \[\pr\left[\arg\max_{i\in I} |X_i| \neq \arg\max_{j\in I} |Y_j|  \right] = O(\sqrt{\epsilon \log k})\]

    Since every $|X_i|$ has the same probability of being the largest, the probability that the largest $|X_i|$ is not in $\{X_i\}_{i\in I}$ is  $(k-|I|)/k$. Similar for that of $|Y_i|$. By a union bound,
    \begin{align*}
        \pr\left[\arg\max_i |X_i| = \arg\max_j |Y_j|\right] &\geq 1-O(\sqrt{\epsilon\log k }) - 2\cdot \frac{k-|I|}{k} 
    \end{align*}
    and we know that 
    \[\frac{k-|I|}{4\log k} \leq \sum_{i\not\in I}\epsilon_i \leq k\epsilon \]
    so we have
    \[\frac{k-|I|}{k}\leq 4\epsilon \log k \leq O(\sqrt{\epsilon \log k})
    \]
    where the last step holds since $\epsilon \log k < 1$. Plug this into the bound above we have 
    \begin{align*}
        \pr\left[\arg\max_i |X_i| = \arg\max_j |Y_j|\right] &\geq 1-O(\sqrt{\epsilon\log k })  
    \end{align*}
    as promised.
\end{proof}

\end{document}